\documentclass[a4paper,english,11pt]{article}

\usepackage{graphicx, amsfonts, amsmath, amssymb, epsfig, color, xparse}
\usepackage{framed}
\usepackage{comment}
\usepackage{caption}

\usepackage[latin9]{inputenc}
\usepackage{babel}
\usepackage{calc}
\usepackage{enumitem}
\usepackage{amsmath}
\usepackage{amsthm}
\usepackage{amssymb}

\usepackage{algorithm,algpseudocode}
\usepackage{amsfonts}
\usepackage{mathtools}

\usepackage{fullpage}

\definecolor{armygreen}{rgb}{0.29, 0.33, 0.13}
\definecolor{darkgreen}{rgb}{0.0, 0.5, 0.0}

\newcommand{\BE}{\begin{enumerate}} \newcommand{\EE}{\end{enumerate}}
\newcommand{\BI}{\begin{itemize}} \newcommand{\EI}{\end{itemize}}
\newcommand{\BDes}{\begin{description}}\newcommand{\EDes}{\end{description}}
\newtheorem{alg}{Algorithm}
\newcommand{\BA}{\begin{alg}} \newcommand{\EA}{\end{alg}}

\newcommand{\BEQ}{\begin{equation}} \newcommand{\EEQ}{\end{equation}}
\newcommand{\BEQN}{\begin{eqnarray}}\newcommand{\EEQN}{\end{eqnarray}}

\newcommand{\bitset}{\{0,1\}}

\newcommand{\set}[1]{\left\{ #1 \right\}}

\newcommand{\todo}[1]{\iffalse {#1} \fi }

\usepackage[colorlinks,citecolor=darkgreen,linkcolor=cyan,bookmarks=true]{hyperref}
\usepackage[nameinlink]{cleveref}

\numberwithin{equation}{section}
\numberwithin{figure}{section}

\newtheorem{theorem}{Theorem}

\newtheorem{lemma}{Lemma}
\newtheorem{claim}[lemma]{Claim}
\newtheorem{observation}[lemma]{Observation}

\newtheorem{corollary}[theorem]{Corollary}

\newtheorem{definition}{Definition}

\theoremstyle{definition}

\makeatletter
\newcommand{\setdef}[1]{%
	\phantomsection
	#1\def\@currentlabel{\unexpanded{#1}}%
}
\makeatother

\makeatletter
\let\ORIGINAL@spythm\@spythm
\def\@spythm#1#2#3#4[#5]{%
	\NR@gettitle{#5}%
	\ORIGINAL@spythm{#1}{#2}{#3}{#4}[#5]%
}
\makeatother

\newcommand{\maj}{{\textrm{MAJ}}}

\newcommand{\cycnums}[1]{{\mathbb{Z}_{#1}}}

\newcommand{\parentheses}[1]{\left(#1\right)}

\title{The Structure of Configurations in One-Dimensional Majority Cellular Automata: \\
	From Cell Stability to Configuration Periodicity}

\author{
	Yonatan Nakar \\
	\small Tel-Aviv University \\
	\small yonatannakar@mail.tau.ac.il
	\and				
	Dana Ron \\
	\small Tel-Aviv University \\
	\small danaron@tau.ac.il
}

\date{}

\begin{document}

\maketitle

\begin{abstract}
We study the dynamics of (synchronous) one-dimensional cellular automata with cyclical boundary conditions that evolve according to the majority rule with radius $ r $.
We introduce a notion that we term \emph{cell stability} with which we express the structure of the possible configurations that could emerge in this setting.
Our main finding is that apart from the configurations of the form $ (0^{r+1}0^* + 1^{r+1}1^*)^* $, which are always fixed-points, the other configurations that the automata could possibly converge to, which are known to be either fixed-points or 2-cycles, have a particular spatially periodic structure.
Namely, each of these configurations is of the form $ s^* $ where $ s $ consists of $ O(r^2) $ consecutive sequences of cells with the same state, each such sequence is of length at most $ r $, and the total length of $ s $ is $ O(r^2) $ as well.
We show that an analogous result also holds for the minority rule.
\end{abstract}

\newpage

\tableofcontents{}

\newpage

\section{Introduction}\label{sec:intro}
Dynamic processes that evolve according to the majority rule arise in various settings and as such have received wide attention in the past, primarily within the context of propagation of information or influence (e.g., \cite{goles2013neural,peleg2002majority,zehmakan2019spread}).
Here we consider perhaps the most basic case, that of one-dimensional cellular automata, where our focus is on analyzing the structure of the configuration space.
Specifically, we analyze the configuration space of one-dimensional cellular automata with cyclical boundary conditions that evolve according to the majority rule with radius $ r $.

It is well-known~\cite{goles1981_period_2,poljak1983period_2} that these processes always converge to configurations that correspond to cycles either of length 1 (\nameref{def:fixed_point}s) or of length 2 (period-2 cycles).
In particular, it is easy to verify (see, e.g., \cite{tosic2004characterizing}) that configurations in which each \ref*{def:cell} belongs to a consecutive sequence of at least $ r+1 $ \ref*{def:cell}s with the same \setdef{state}\label{def:state}\footnote{In this work, a \ref*{def:state} is a value in $ \bitset $.} are \nameref{def:fixed_point}s.
Not much is currently understood, however, about the structure of the other \nameref{def:fixed_point} configurations or of configurations that correspond to cycles of length 2.

The reason for this gap in understanding is largely due to the fact that most previous research has made assumptions about the mechanism producing the initial configuration.
Namely, it is usually assumed that the \ref*{def:state} of each \ref*{def:cell} in the initial configuration is randomly chosen, independently from the other \ref*{def:cell}s.
See, for instance, the theoretical analysis in \cite{tosic2004characterizing} and the experimental results in \cite{tosic2011convergence_properties_1_dimensional_majority}, both for one-dimensional majority cellular automata (and also the references within \Cref{sec:related_work} for examples in other models).
Under such assumptions, as shown in \cite{tosic2004characterizing}, these other configurations are indeed rarely encountered.

In this work, we tackle the problem of understanding the structure of the possible configurations without making assumptions about the mechanism behind the generation of the initial configuration.
One of our main results (stated formally in \Cref{thm:characterization}) is that all period-2 configurations and all \nameref{def:fixed_point} configurations (other than those mentioned above) have a very special structure.
Specifically, they have a ``spatially'' periodic structure with a period that is quadratic in the radius $ r $.
In the course of the proof of this result, we introduce several notions and prove several claims, which we believe are of interest in their own right as they shed light on the dynamics of the majority rule in cellular automata (and not only on the configurations they converge to).

\subsection{The majority rule with radius \texorpdfstring{$ r $}{r}}\label{sec:majority_rule}

In all that follows, when performing operations on \setdef{cell}s\label{def:cell} $ i \in \cycnums{n} $, these operations are modulo $ n $.

\begin{definition}[cell interval]\label{def:cell_interval}
	For a pair of \ref*{def:cell}s $ i,j \in \cycnums{n} $ we use $ [i,j] $ to denote the sequence $ i,i+1,\dots,j $ (so that it is possible that $ j < i $), which we refer to as a \textsf{\nameref{def:cell_interval}}.
\end{definition}

For an integer $ n $, we refer to a function $ \sigma : \cycnums{n} \to \bitset $ as a \emph{configuration} and view $ \sigma $ as a (cyclic) binary string of length $ n $.

\begin{definition}[neighborhood]\label{def:neighborhood}
	For a \ref*{def:cell} $ i\in \cycnums{n} $ and an integer $ r $, the $ r $-\textsf{\nameref{def:neighborhood}} of $ i $, denoted $ \Gamma_r(i) $, is the \nameref{def:cell_interval} $ [i-r,i+r] $.
	For a set of \ref*{def:cell}s $ I \subseteq \cycnums{n} $, we let $ \Gamma_r(I) $ denote the set of \ref*{def:cell}s in the union of \nameref{def:cell_interval}s $ [i-r,i+r] $ taken over all $ i\in I $.
\end{definition}

Given a \ref*{def:state} $ \beta \in \set{0, 1} $, a configuration $ \sigma : \cycnums{n} \to \bitset $ and a \nameref{def:cell_interval} $ [i,j] $, we denote by $ \#_\beta(\sigma[i,j]) $ the number of \ref*{def:cell}s $ \ell \in [i,j] $ such that $ \sigma(\ell)=\beta $.

\begin{definition}[the majority rule]\label{def:majority_rule}
	Denote by $ \maj_r $ \textsf{\nameref*{def:majority_rule} with radius $ r $}.
	That is, for a configuration $ \sigma : \cycnums{n} \to \bitset $, $ \maj_r(\sigma) $ is the configuration $ \sigma' $ in which for each \ref*{def:cell} $ i \in \cycnums{n} $,
	\begin{align*}
		\sigma'(i) =
		\begin{cases}
			0 & \text{if } \#_0(\sigma[\Gamma_r(i)]) > \#_1(\sigma[\Gamma_r(i)]) \\
			1 & \text{otherwise}
		\end{cases}
	\end{align*}
\end{definition}

For each $ t \ge 0 $, denote by $ \maj_r^t(\sigma) $ the result of repeatedly applying the majority rule with radius $ r $, starting from the configuration $ \sigma $.
In particular, $ \maj_r^0(\sigma) = \sigma $ and $ \maj_r^1(\sigma) = \maj_r(\sigma) $.

\subsection{Temporal and spatial periodicity}
Eventually, for every initial configuration, the majority rule, and, in fact, any rule, reaches a \emph{cycle}: a periodic sequence of configurations.
As mentioned earlier, in the case of the majority rule, that cycle is always either a \emph{\nameref{def:2_cycle}} or a \emph{\nameref{def:fixed_point}}.

\begin{definition}[fixed-point]\label{def:fixed_point}
	We say that a configuration $ \sigma : \cycnums{n} \to \bitset $ is a \textsf{\nameref{def:fixed_point}} if $ \maj_r(\sigma) = \sigma $.
\end{definition}

\begin{definition}[2-cycle]\label{def:2_cycle}
	We say that a pair of distinct configurations $ \sigma, \sigma' : \cycnums{n} \to \bitset $ is a \textsf{\nameref{def:2_cycle}} if $ \maj_r(\sigma)=\sigma' $ and $ \maj_r(\sigma')=\sigma $.
\end{definition}

We refer to the \emph{configurations} that constitute a cycle as \nameref{def:temporally_periodic} configurations.
That is,

\begin{definition}[temporally periodic]\label{def:temporally_periodic}
	We say that a configuration $ \sigma : \cycnums{n} \to \bitset $ is \textsf{\nameref{def:temporally_periodic}} if $ \maj_r^2(\sigma) = \sigma $.
\end{definition}

Note that if a configuration $ \sigma $ is \nameref{def:temporally_periodic}, then it is either the case that $ \maj_r(\sigma)=\sigma $ (i.e., $ \sigma $ is a \nameref{def:fixed_point}), or $ \maj_r(\sigma)=\sigma' $ for $ \sigma' \ne \sigma $, in which case $ \sigma $ and $ \sigma' $ constitute a \nameref{def:2_cycle}.

\begin{definition}[transient]\label{def:transient}
	If a configuration $ \sigma : \cycnums{n} \to \bitset $ is not \nameref{def:temporally_periodic}, we say that $ \sigma $ is \textsf{\nameref{def:transient}}.
\end{definition}

\nameCrefs{def:2_cycle} \ref{def:fixed_point}-\ref{def:transient} are all related to the notion of \emph{temporal} periodicity, i.e., periodicity that occurs over time.
In this paper, we relate temporal periodicity to \emph{spatial} periodicity, i.e., periodic behavior exhibited within individual configurations.
Formally,

\begin{definition}[spatial period]\label{def:spatial_period}
	We say that a configuration $ \sigma : \cycnums{n} \to \bitset $ has \textsf{\nameref{def:spatial_period}} $ p $ if $ p $ is the minimum positive integer such that for every \ref*{def:cell} $ i \in \cycnums{n} $,
	$ \sigma(i+p) = \sigma(i) . $
\end{definition}

\begin{definition}[spatially periodic]\label{def:spatially_periodic}
	We say that a configuration $ \sigma : \cycnums{n} \to \bitset $ is \textsf{\nameref{def:spatially_periodic}} if its \nameref{def:spatial_period} $ p $ satisfies $ p < n $.
\end{definition}

\subsection{Our main result and the notion of \ref*{def:cell} stability}\label{sec:result}
In this section we state our main result, \Cref{thm:characterization}.
In order to state \Cref{thm:characterization}, we introduce the notion of a \ref{def:cell}'s stability within a configuration via \nameCrefs{def:unstable}~\ref{def:unstable}-\ref{def:weakly_stable} (illustrated in \Cref{fig:majority_annotated}).

\begin{definition}[unstable]\label{def:unstable}
	We say that a \ref*{def:cell} $ i \in \cycnums{n} $ is \textsf{\nameref{def:unstable}} with respect to a configuration $ \sigma : \cycnums{n} \to \bitset $ if $ \sigma(i) \ne \sigma''(i) $ where $ \sigma'' = \maj_r^2(\sigma) $.
\end{definition}

Recall that after a finite number of steps\footnote{Which is shown in \Cref{sec:temporally_periodic_configurations} to be at most linear in $ n $.}, a one-dimensional cellular automaton that evolves according to the majority rule, reaches either a \nameref{def:fixed_point} or a \nameref{def:2_cycle}.
Thus, a configuration $ \sigma : \cycnums{n} \to \bitset $ is \nameref{def:transient} if and only if it contains \nameref{def:unstable} \ref*{def:cell}s.

As for the ``stable'' \ref*{def:cell}s, we define two variants: \nameref{def:strongly_stable} and \nameref{def:weakly_stable}.

\begin{definition}[strongly stable]\label{def:strongly_stable}
	We say that a \ref*{def:cell} $ i \in \cycnums{n} $ is \textsf{\nameref{def:strongly_stable}} with respect to a configuration $ \sigma : \cycnums{n} \to \bitset $ if there exists a \nameref{def:cell_interval} $ [a,b] $ of length at least $ r+1 $ such that $ i \in [a,b] $ and for each $ j \in [a,b] $, $ \sigma(i) = \sigma(j) $.
\end{definition}

\begin{definition}[weakly stable]\label{def:weakly_stable}
	We say that a \ref*{def:cell} $ i \in \cycnums{n} $ is \textsf{\nameref{def:weakly_stable}} with respect to a configuration $ \sigma : \cycnums{n} \to \bitset $ if $ i $ is not \nameref{def:strongly_stable} with respect to $ \sigma $, but $ \sigma(i) = \sigma''(i) $ where $ \sigma'' = \maj_r^2(\sigma) $.
\end{definition}

\begin{figure}[htb!]
	\begin{framed}
		\centerline{\mbox{\includegraphics[width=1\textwidth]{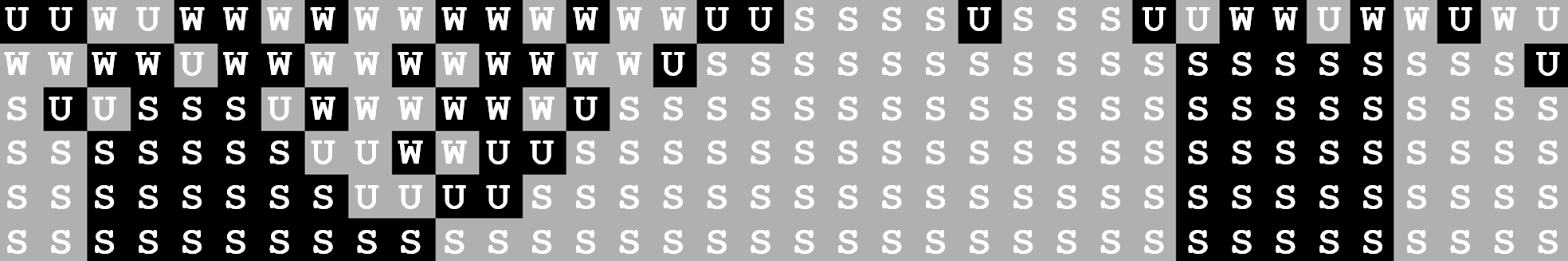}}}
		\caption{
			\small
			The evolution under the majority rule with $ r=2 $.
			Gray squares correspond to \ref*{def:state}-0 \ref*{def:cell}s and dark squares correspond to \ref*{def:state}-1 \ref*{def:cell}s.
			Each \ref*{def:cell} is labeled by a letter indicating the \ref*{def:cell}'s \emph{stability}, where \emph{S} stands for \emph{Strongly} stable, \emph{W} for \emph{Weakly} stable and \emph{U} for \emph{Unstable}.
		}
		\label{fig:majority_annotated}
	\end{framed}
\end{figure}

The crucial property of the \nameref{def:strongly_stable} \ref*{def:cell}s is that their \ref*{def:state}s, unlike the \ref*{def:state}s of the \nameref{def:weakly_stable} \ref*{def:cell}s, cannot change in later configurations.
In that sense, their stability is ``stronger'' than that of the \nameref{def:weakly_stable} \ref*{def:cell}s.
It is worth noting, though, that if a \ref*{def:cell} lies within a long \nameref{def:cell_interval} of \nameref{def:weakly_stable} \ref*{def:cell}s, then that \ref*{def:cell} remains \nameref{def:weakly_stable}, alternating between the same pair of \ref*{def:state}s, for a number of steps that depends on the \nameref{def:cell_interval} length.

Accordingly, given a configuration $ \sigma : \cycnums{n} \to \bitset $, we say that a \nameref{def:cell_interval} $ [i,j] $ is \nameref{def:strongly_stable}, \nameref{def:weakly_stable} or \nameref{def:unstable} if all the \ref*{def:cell}s in that \nameref{def:cell_interval} are, respectively, \nameref{def:strongly_stable}, \nameref{def:weakly_stable} or \nameref{def:unstable}.

Considering complete configurations, observe that all the configurations of the form $ (0^{r+1}0^* + 1^{r+1}1^*)^* $ contain only \nameref{def:strongly_stable} \ref*{def:cell}s.
As noted previously and explained in the characterization provided in \cite{tosic2004characterizing}, these configurations are always \nameref{def:fixed_point}s, which means that they are, in particular, also \nameref{def:temporally_periodic} (with a period of 1).
However, there are more forms of \nameref{def:temporally_periodic} configurations, both period-1 and period-2, that contain only \nameref{def:weakly_stable} \ref*{def:cell}s and are not addressed by \cite{tosic2004characterizing}'s characterization, as the authors of \cite{tosic2004characterizing} were only interested in ``typical'' configurations, which are not of that kind.\footnote{
	Indeed, it is shown in \cite{tosic2004characterizing} that the probability that a randomly selected configuration of length $ n $ being \nameref{def:transient} approaches 1 as $ n \longrightarrow \infty $. As such, the additional \nameref{def:temporally_periodic} configurations that we address in this work are, in a sense, not ``typical''.
	We, in contrast to \cite{tosic2004characterizing}, make no assumption about the distribution of the configuration space, and are therefore interested in understanding the structure of \emph{all} configurations, not only the ``typical'' ones.
}

\Cref{thm:characterization} complements \cite{tosic2004characterizing}'s characterization by additionally specifying the structure of the remaining \nameref{def:temporally_periodic} configurations.
In addition to \nameref{def:temporally_periodic} configurations, \Cref{thm:characterization} also includes a property of the \nameref{def:transient} configurations that is related to the dynamics by which they eventually converge.

\begin{theorem}\label{thm:characterization}
	For any configuration $ \sigma : \cycnums{n} \to \bitset $, exactly one of the following must hold:
	\begin{enumerate}
		
		\item The configuration $ \sigma $ is a \nameref{def:temporally_periodic} configuration and it is either the case that:
		\begin{enumerate}
			\item\label{case:strongly_stable} all the \ref*{def:cell}s in $ \sigma $ are \nameref{def:strongly_stable}, in which case $ \sigma $ is of the form $ (0^{r+1}0^* + 1^{r+1}1^*)^* $), or
			\item\label{case:weakly_stable} all the \ref*{def:cell}s in $ \sigma $ are \nameref{def:weakly_stable}, in which case $ \sigma $ is \nameref{def:spatially_periodic} with \nameref{def:spatial_period} at most $ 2r(r+1) $.
		\end{enumerate}
		
		\item\label{case:unstable} The configuration $ \sigma $ is a \nameref{def:transient} configuration and the length of every \nameref{def:unstable} \nameref{def:cell_interval} in $ \sigma $ is at most $ 2r $.
	\end{enumerate}
\end{theorem}

Under the assumption that $ r $ is a constant, \Cref{thm:characterization} directly yields an output-sensitive algorithm that, given $ n $, generates all the \nameref{def:temporally_periodic} configurations of length $ n $.
The running-time of the algorithm is linear in the number of \nameref{def:temporally_periodic} configurations.

Turning to \nameref{def:transient} configurations, recall that all \nameref{def:transient} configurations contain \nameref{def:unstable} \ref*{def:cell}s, and the evolution of the \nameref{def:transient} configurations can be described using the notion of \ref*{def:cell} stability.
Namely, the following is shown (in \Cref{sec:temporally_periodic_configurations}) regarding any \nameref{def:transient} configuration $ \sigma : \cycnums{n} \to \bitset $.
First, the configuration $ \maj_r(\sigma) $ contains strictly fewer \nameref{def:unstable} \ref*{def:cell}s than $ \sigma $.
Second, if $ \sigma $ contains \nameref{def:strongly_stable} \ref*{def:cell}s, then $ \maj_r(\sigma) $ contains even more \nameref{def:strongly_stable} \ref*{def:cell}s than $ \sigma $, and the automaton eventually converges to a \nameref{def:fixed_point} of the form defined in Case~(\ref{case:strongly_stable}).
Third, if there are no \nameref{def:strongly_stable} \ref*{def:cell}s in $ \sigma $, then there are cases in which the automaton eventually converges to a \nameref{def:fixed_point} of the form defined in Case~(\ref{case:strongly_stable})\footnote{
	e.g., for $ r=3 $, the \nameref{def:transient} configuration $ 001001001001001001 $ converges after one step to the \nameref{def:fixed_point} configuration $ (0)^* $.
}
and there are also cases in which it eventually converges to a \nameref{def:fixed_point} or to a \nameref{def:2_cycle} of the form defined in Case~(\ref{case:weakly_stable})\footnote{
	e.g., for $ r=4 $, the \nameref{def:transient} configuration $ 001011001011001011001011001011001011 $ converges after one step to the \nameref{def:2_cycle} consisting of $ (111000)^6 $ and $ (000111)^6 $.
}.

\subsection{Illustrating \Cref*{thm:characterization} for \texorpdfstring{ $ r = 1, 2, 3 $}{r = 1, 2, 3}}\label{sec:thm_illustration}
To get a feel for the nature of the statement in \Cref{thm:characterization}, we demonstrate some of its aspects for $ r = 1, 2,  3 $.

\begin{enumerate}
	\item For $ r=1 $, the \nameref{def:temporally_periodic} configurations are either
	\begin{enumerate}
		\item of the form $ (000^* + 111^*)^* $, or
		\item of the form $ (01)^* $.\footnote{Also $ (10)^* $, but since the configurations are cyclic, the patterns $ (01)^* $ and $ (10)^* $ correspond to equivalent sets of configurations.}
	\end{enumerate}
	
	\item For $ r=2 $, the \nameref{def:temporally_periodic} configurations are either
	\begin{enumerate}
		\item of the form $ (0000^* + 1111^*)^* $, or
		\item of one of the following forms: $ (01)^* $, $ (0011)^* $, $ (001101)^* $, $ (001011)^* $.
	\end{enumerate}
	
	\item For $ r=3 $, the \nameref{def:temporally_periodic} configurations are either
	\begin{enumerate}
		\item of the form $ (00000^* + 11111^*)^* $, or
		\item of the form $ s^* $, where $ s $ belongs to the set:\footnote{
			The string $ s $ could also be the \emph{mirror} or the \emph{complement} of any of the specified patterns, which we omit for the sake of conciseness.
			For example, since we explicitly specified that $ s $ could be $ 010011 $, it means that $ s $ could also be $ 110010 $ (which is the mirror of $ 010011 $) or $ 101100 $ (which is the complement of $ 010011 $), even though these two are not explicitly specified.
		}
	\end{enumerate}
	\begin{align*}
		\left\{
		\begin{array}{lr}
			01, \\
			0011, \\
			010011, 010110, 001110, \\
			01011001, 10100101, 10100110, 01011100, 10010011, 00011101, 10110001, \\
			0011001110, 1000111001
		\end{array}
		\right\}
	\end{align*}
\end{enumerate}

\subsection{Minority}\label{sec:minority}
An analog of \Cref{thm:characterization} holds for the minority rule as well, with exactly the same variants of \ref{def:cell} stability as those of the majority rule.
In particular, \nameCrefs{def:unstable}~\ref{def:unstable}-\ref{def:weakly_stable} can be used verbatim to describe the evolution according to the minority rule, with the only difference being the \emph{temporal period} of the \nameref{def:weakly_stable} and the \nameref{def:strongly_stable} \ref{def:cell}s.

Namely, in the minority rule, the \nameref{def:strongly_stable} \ref{def:cell}s have a temporal period of 2 instead of 1, implying that the \nameref{def:temporally_periodic} configurations of the form defined in Case~(\ref{case:strongly_stable}) of \Cref{thm:characterization}, rather than being \nameref{def:fixed_point}s as in the majority rule, become the constituents of \nameref{def:2_cycle}s.
Likewise, for every configuration, the \nameref{def:weakly_stable} \ref{def:cell}s that would have had a temporal period of 1 under the majority rule, have a temporal period of 2 under the minority rule, and vice versa.
This implies that the \nameref{def:temporally_periodic} configurations of the form defined in Case~(\ref{case:weakly_stable}) of \Cref{thm:characterization}, while generated by exactly the same patterns, and hence having precisely the same form, have the \emph{opposite} temporal period to that they would have had under the majority rule: the \nameref{def:fixed_point}s become \nameref{def:2_cycle}s and the \nameref{def:2_cycle}s become \nameref{def:fixed_point}s.

See illustrations in \nameCrefs{fig:majority} \ref{fig:majority} and \ref{fig:minority}.

\begin{figure}[htb!]
\begin{framed}
	\begin{framed}
	\centerline{\mbox{\includegraphics[width=1\textwidth]{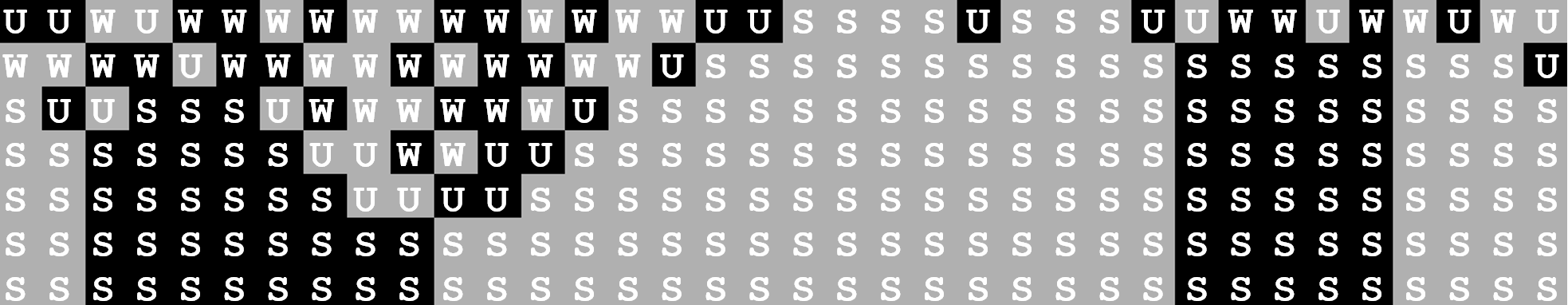}}}
	\caption{\small Majority ($ r=2 $)}
	\label{fig:majority}
\end{framed}
\begin{framed}
	\centerline{\mbox{\includegraphics[width=1\textwidth]{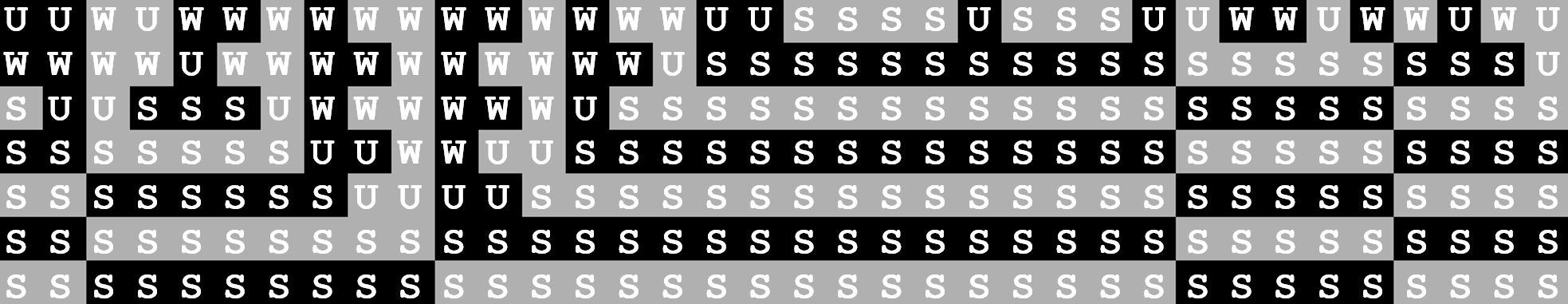}}}
	\caption{\small Minority ($ r=2 $)}
	\label{fig:minority}
\end{framed}
\captionsetup{labelformat=empty}
\caption{
	\small
	\nameCrefs{fig:majority} \ref{fig:majority} and \ref{fig:minority} depict the evolution under the majority rule (\Cref{fig:majority}) and under the minority rule (\Cref{fig:minority}) starting from the same initial configuration.
	Gray squares correspond to \ref*{def:state}-0 \ref{def:cell}s and dark squares correspond to \ref*{def:state}-1 \ref{def:cell}s.
	Each \ref{def:cell} is labeled by a letter indicating the \ref{def:cell}'s \emph{stability}, where \emph{S} stands for \emph{Strongly} stable, \emph{W} for \emph{Weakly} stable and \emph{U} for \emph{Unstable}.
	While the two trajectories appear quite different, they are completely equivalent when viewed under the ``stability mask''.
	In particular, \ref{def:cell} $ i $ at time $ t $ has the same stability label in both figures for every $ t, i $ pair.
	Yet, each pair of corresponding stable \ref{def:cell}s in the two figures have \emph{opposite} periods.
	Thus, the \nameref{def:fixed_point} configuration to which the majority rule in \Cref{fig:majority} converges to corresponds to a \nameref{def:2_cycle} in \Cref{fig:minority}.
}
\end{framed}
\end{figure}

\subsection{Related work}\label{sec:related_work}
The main focus of most of the research on majority/minority (and more generally, threshold) cellular automata so far has been on the convergence time~(e.g., \cite{fogelman1983convergence,frischknecht2013convergence,papp2019stabilization}) and on the dominance problem\footnote{In the dominance problem, one asks how many \ref*{def:cell}s must initially be at a certain \ref*{def:state} so that eventually all \ref*{def:cell}s have the same \ref*{def:state}.} (e.g., \cite{balogh2012sharp,flocchini2003time,mitsche2017strong}).

As mentioned earlier, most of the work on the problem of understanding the structure of the configuration space is based on the assumption that the initial configuration is random.
For the one-dimensional case, the case with which the current paper is concerned, this includes the paper of Tosic and Agha~\cite{tosic2004characterizing}.
In their paper, they distinguish between synchronous/sequential and finite/infinite majority cellular automata with radius $ r $, and our work can be viewed as extending their result for the finite and synchronous case.

They show that whereas \nameref{def:2_cycle}s cannot emerge under the sequential model, in the synchronous model (the one we focus on in this paper), \nameref{def:2_cycle}s exist even for $ r=1 $.
They also show that a randomly picked configuration is a \nameref{def:transient} configuration (and, in particular, not a \nameref{def:2_cycle}) with probability approaching 1 (both for finite and infinite configurations), and it can additionally be shown that the probability that such a random \nameref{def:transient} configuration eventually converges to a \nameref{def:2_cycle} approaches 0.
Finally, they characterize the ``common'' forms of \nameref{def:fixed_point} configurations (those that in our paper are described in Case~(\ref{case:strongly_stable}) of \Cref{thm:characterization}).

Their theoretical result is supplemented by a later experimental work~\cite{tosic2011convergence_properties_1_dimensional_majority}, showing that in practice, convergence to these ``common'' \nameref{def:fixed_point} configurations occurs relatively quickly.
Namely, the simulations in \cite{tosic2011convergence_properties_1_dimensional_majority} demonstrate that convergence tends to occur in less than five steps for $ n=1000 $ and $ 1 \le r \le 5 $.

Additional work beyond the one-dimensional case includes \cite{Zehmakan2021majority_2d_cellular_automata} for two-dimensional \sloppy majority cellular automata, \cite{Zehmakan2018majority} for majority in random regular graphs, \cite{zehmakan2020opinion} for majority in Erdos--R{\'e}nyi graphs as well as expander graphs.

One notable work that does not rely on the assumption that the initial configuration is random is Turau's work~\cite{turau2022trees} on characterizing all the \nameref{def:temporally_periodic} configurations for majority and minority processes on trees.
The characterization presented in \cite{turau2022trees} also yields an output-sensitive algorithm for generating these configurations.

\subsection{Some high-level ideas}\label{sec:high_level_ideas}
As mentioned previously, in proving \Cref{thm:characterization}, we define a number of notions and establish several claims, some of which we believe are valuable in and of themselves.
In this section we have chosen to highlight the high-level idea behind one of the key tools we utilize, which is a \emph{mapping} we introduce between \emph{\ref*{def:block}s} of consecutive configurations.

Given a configuration $ \sigma : \cycnums{n} \to \bitset $, we say that a \nameref{def:cell_interval} $ [i,j] $ is a \emph{\setdef{maximal homogeneous block}\label{def:block}} in $ \sigma $ with value $ \beta \in \bitset $ if for every \ref*{def:cell} $ \ell \in [i,j] $, $ \sigma(\ell)=\beta $, and also $ \sigma(i-1) = \sigma(j+1) \ne \beta $ if the length of $ [i,j] $ is less than $ n $.

We refer to this mapping, defined below (and illustrated in \Cref{fig:alignment}), as the \nameref{def:alignment_mapping}.
The \nameref{def:alignment_mapping}, beyond being essential for the proof of \Cref{thm:characterization}, has several features that make it useful for reasoning about the dynamics of the majority rule, which is why we present its definition here.

\begin{definition}[alignment mapping]
	Let $ \sigma $ and $ \sigma' $ be a pair of configurations satisfying $ \maj_r(\sigma) = \sigma' $.
	Given a \ref*{def:block} $ [i',j'] $ in $ \sigma' $, let $ I $ be the \ref*{def:block} in $ \sigma $ that contains the \ref*{def:cell} $ i+r $ and let $ J $ be the \ref*{def:block} in $ \sigma $ that contains the \ref*{def:cell} $ j-r $.
	\textsf{The \nameref*{def:alignment_mapping}} maps the \ref*{def:block} $ [i',j'] $ (in $ \sigma' $) to the \emph{middle}\footnote{The middle block is well defined, as it is shown in \Cref{sec:left_right_mappings} that the number of \ref*{def:block}s between $ I $ and $ J $ must be odd.} \ref*{def:block} $ [i,j] $ between $ I $ and $ J $ in $ \sigma $.
\end{definition}

\begin{figure}[htb!]
	\begin{framed}
		\centerline{\mbox{\includegraphics[width=1\textwidth]{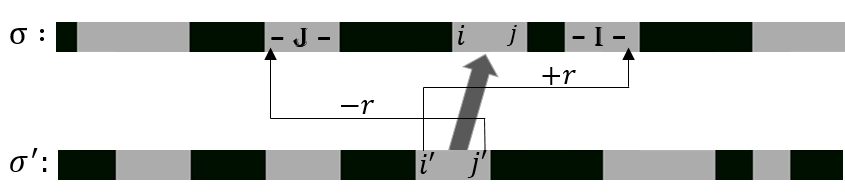}}}
		\caption{
			\small
			\emph{The \nameref*{def:alignment_mapping}.}
			The figure depicts a pair of configurations, $ \sigma $ and $ \sigma' $, where $ \sigma' = \maj_r(\sigma) $, and also a pair of \ref*{def:block}s, $ [i,j] $ in $ \sigma $ and $ [i',j'] $ in $ \sigma' $, where $ [i',j'] $ is mapped to $ [i,j] $ by the \nameref{def:alignment_mapping}.
			The \ref*{def:block} $ I $ in $ \sigma $ is the one that contains the \ref*{def:cell} $ i'+r $, and the \ref*{def:block} $ J $ in $ \sigma $ is the one that contains the \ref*{def:cell} $ j'-r $.
			The \ref*{def:block} $ [i,j] $ in $ \sigma $ is the one right in the middle of the interval of five \ref*{def:block}s in $ \sigma $ whose left and right ends are $ I $ and $ J $.
			Hence, by definition, the \nameref{def:alignment_mapping} maps $ [i',j'] $ to $ [i,j] $.
		}
		\label{fig:alignment}
	\end{framed}
\end{figure}

We stress that the \nameref{def:alignment_mapping}, as defined above (as well as in \Cref{sec:alignment_mapping}), is a \emph{backward} mapping, in the sense that, given a configuration $ \sigma' $, it maps all \ref*{def:block}s in $ \sigma' $ into those of the configuration $ \sigma $ that \emph{precedes} $ \sigma' $.
This naturally suggests defining the notion of \emph{the forward \nameref{def:alignment_mapping}} as the \emph{inverse} function of the backward \nameref{def:alignment_mapping} that would map the \ref*{def:block}s of the configuration $ \sigma $ to those of the configuration $ \sigma' $ that \emph{follows} $ \sigma $ (for example, in \Cref{fig:alignment}, the forward \nameref{def:alignment_mapping} maps $ [i,j] $ in $ \sigma $ to $ [i',j'] $ in $ \sigma' $).

However, while it can be shown that the backward \nameref{def:alignment_mapping} is always one-to-one, it is not necessarily \emph{onto} (unless we apply it within a pair of \nameref{def:temporally_periodic} configurations).
Hence, under our definition of the forward \nameref{def:alignment_mapping}, not all blocks will be mapped forward.

Formally, let $ \sigma_0,...\sigma_m $ be a sequence of configurations where $ \maj_r(\sigma_{t-1}) = \sigma_t $ for each $ 1 \le t \le m $.
We define the step-$ t $ \emph{forward \nameref{def:alignment_mapping}}, denoted $ \varphi_t $, as follows.
Given a \ref*{def:block} $ [i,j] $ in $ \sigma_t $, if there is a \ref*{def:block} $ [i',j'] $ in $ \sigma_{t+1} $ such that the backward \nameref{def:alignment_mapping} between the configuration pair $ \sigma_t, \sigma_{t+1} $ maps $ [i',j'] $ into $ [i,j] $, then $ \varphi_t([i,j]) = [i',j'] $.
Otherwise, $ \varphi_t([i,j]) = \bot $.
In the case in which $ \varphi_t([i,j]) \ne \bot $, we also define $ \varphi_t^2([i,j]) $ as $ \varphi_{t+1}(\varphi_t([i,j])) $.

One notable property of the forward \nameref{def:alignment_mapping} is what we refer to as ``identity preservation in stable intervals''.
Roughly speaking, consider any \ref*{def:block} $ [i,j] $ residing in a sufficiently long \nameref{def:weakly_stable} or \nameref{def:strongly_stable} \nameref{def:cell_interval} of $ \sigma_t $.
Then $ \varphi_t([i,j]) \ne \bot $, and hence $ \varphi_t^2([i,j]) $ is defined and is equal to the same \ref*{def:block} $ [i,j] $ we started with.
In particular, for a pair of configurations comprising a \nameref{def:2_cycle}, applying the forward \nameref{def:alignment_mapping} \emph{twice} essentially maps each \ref*{def:block} to itself.

In the proof of \Cref{thm:characterization}, we essentially use the forward \nameref{def:alignment_mapping} and its properties to show that for a configuration in which all \ref*{def:block}s are of length at most $ r $, if the configuration is \nameref{def:temporally_periodic}, then it is also \nameref{def:spatially_periodic}.
We achieve this through three steps.

In the first step (\nameCrefs{sec:diff_vectors} \ref{sec:block_lengths}-\ref{sec:block_length_vectors}), we employ the \nameref{def:alignment_mapping} to express the length of each of the configuration's \ref*{def:block}s in terms of the lengths of other $ O(r) $ \ref*{def:block}s in the preceding configuration.
Specifically, given a pair of \nameref{def:temporally_periodic} configurations $ \sigma_t $ and $ \sigma_{t+1} $, we obtain a relationship between the length of each \ref*{def:block} $ [i,j] $ in $ \sigma_t $ and the lengths of $ O(r) $ consecutive \ref*{def:block}s, belonging to a block sequence centered at the \ref*{def:block} $ \varphi_t([i,j]) $, in the configuration $ \sigma_{t+1} $.

In the second step (\nameCrefs{sec:diff_vectors} \ref{sec:horizon}-\ref{sec:diff_vectors}), we look at the \emph{difference} between the length of each \ref*{def:block} $ [i,j] $ and the lengths of the \ref*{def:block}s at the two ends of the sequence mentioned above, and define \emph{aligned difference vectors}, whose entries are these differences.
We use the properties of the forward \nameref{def:alignment_mapping} to establish that the aligned difference vectors (defined formally in \Cref{sec:diff_vectors}) are \nameref{def:spatially_periodic} with a \nameref{def:spatial_period} that is \emph{linear} in $ r $.

In the third and final step (\Cref{sec:spatially_periodic}), by applying the relationship between aligned difference vectors iteratively, we use the spatial periodicity of the aligned difference vectors to establish that the configurations themselves are \nameref{def:spatially_periodic} as well, and that each configuration's \nameref{def:spatial_period} must be quadratic in $ r $.

\section{Switch Points}

For a value $ \beta \in \bitset $, we denote $ \bar{\beta} = 1-\beta $.

\begin{definition}[switch point]\label{def:switch_point}
	Let $ \sigma $ be a configuration. We say that a pair of consecutive \ref{def:cell}s $ i,i+1 $ constitute a \nameref{def:switch_point} in $ \sigma $ if $ \sigma(i) \ne \sigma(i+1) $.
\end{definition}

We now prove two properties of \nameref{def:switch_point}s, the first of which we refer to as \nameref{claim:switch_point_arg}.

\begin{claim}[the switch point property]\label{claim:switch_point_arg}
	Let $ \sigma $ and $ \sigma' $ be a pair of configurations satisfying $ \maj_r(\sigma) = \sigma' $.
	If a pair of consecutive \ref{def:cell}s $ i,i+1 $ constitutes a \nameref{def:switch_point} in $ \sigma' $, then 
	$ \sigma(i-r) = \sigma'(i) $ and $ \sigma(i+1+r) = \sigma'(i+1) $.
\end{claim}

\begin{proof}
	We only prove that $ \sigma(i-r) = \sigma'(i) $, since the proof that $ \sigma(i+1+r) = \sigma'(i+1) $ is symmetric.
	
	Suppose $ \sigma'(i+1) = \beta $ for some $ \beta \in \set{0,1} $.
	Since the pair $ i,i+1 $ constitutes a \nameref{def:switch_point} in $ \sigma' $, it must be the case that $ \sigma'(i) = \bar{\beta} $.
	Assume, contrary to the claim, that $ \sigma(i-r) = \beta $.
	Since $ \Gamma_r(i+1) = \Gamma_r(i) \cup \set{i+1+r} \setminus \set{i-r} $ and $ \sigma(i-r) = \beta $, it must hold that
	\begin{align*}
		\#_{\beta}(\sigma[\Gamma_r(i+1)]) \le \#_{\beta}(\sigma[\Gamma_r(i)]) .
	\end{align*}
	This implies that, by the definition of $ \maj_r $, since $ \sigma'(i+1) = \beta $, it must also hold that $ \sigma'(i) = \beta $, and we reach a contradiction.
\end{proof}

\begin{claim}\label{claim:switch-point-induces-balance}
	Let $ \sigma $ and $ \sigma' $ be a pair of configurations where $ \maj_r(\sigma) = \sigma' $.
	If a pair of consecutive \ref{def:cell}s $ i,i+1 $ constitutes a \nameref{def:switch_point} in $ \sigma' $, then
	\begin{align*}
		\#_0(\sigma[\Gamma_r(i) \cap \Gamma_r(i+1)]) = \#_1(\sigma[\Gamma_r(i) \cap \Gamma_r(i+1)]) .
	\end{align*}
\end{claim}
\begin{proof}
	Assume the contrary, and let $ \beta \in \set{0,1} $ be the majority value in $ \sigma[\Gamma_r(i) \cap \Gamma_r(i+1)] $.
	That is,
	\begin{align*}
		\#_{\beta}(\sigma[\Gamma_r(i) \cap \Gamma_r(i+1)]) > \#_{\bar{\beta}}(\sigma[\Gamma_r(i) \cap \Gamma_r(i+1)]) .
	\end{align*}
	Hence, since $ |\Gamma_r(i) \cap \Gamma_r(i+1)| = |[i-r+1,i+r]| = 2r $,
	\begin{align*}
		\#_{\beta}(\sigma[\Gamma_r(i) \cap \Gamma_r(i+1)]) \ge r+1 .
	\end{align*}
	Thus, since $ \Gamma_r(i) \cap \Gamma_r(i+1) \subseteq \Gamma_r(i) $ and $ \Gamma_r(i+1) \cap \Gamma_r(i) \subseteq \Gamma_r(i+1) $,
	\begin{align*}
		\#_{\beta}(\Gamma_r(i)) \ge \#_{\beta}(\sigma[\Gamma_r(i) \cap \Gamma_r(i+1)]) \ge r+1
	\end{align*}
	and
	\begin{align*}
		\#_{\beta}(\Gamma_r(i+1)) \ge \#_{\beta}(\sigma[\Gamma_r(i) \cap \Gamma_r(i+1)]) \ge r+1 .
	\end{align*}
	Therefore, by the definition of $ \maj_r $,
	$ \sigma(i) = \sigma(i+1) = \beta $ and we reach a contradiction to the assumption that the pair $ i,i+1 $ is a \nameref{def:switch_point}.
\end{proof}

The two properties of \nameref{def:switch_point}s captured by \Cref{claim:switch_point_arg} and \Cref{claim:switch-point-induces-balance} together correspond, in fact, to a characterization of \nameref{def:switch_point}s, as formalized in \Cref{claim:claim:switch_point_arg-inverse}.

\begin{claim}\label{claim:claim:switch_point_arg-inverse}
	Let $ \sigma $ and $ \sigma' $ be a pair of configurations where $ \maj_r(\sigma) = \sigma' $.
	If for a \ref{def:cell} $ i \in \cycnums{n} $, 
	$
	\#_0(\sigma[\Gamma_r(i) \cap \Gamma_r(i+1)]) = \#_1(\sigma[\Gamma_r(i) \cap \Gamma_r(i+1)]) 
	$
	and $ \sigma(i-r) \ne \sigma(i+1+r) $,
	then the pair $ i,i+1 $ constitutes a \nameref{def:switch_point} in $ \sigma' $ in which
	$ \sigma'(i) = \sigma(i-r) $ and $ \sigma'(i+1) = \sigma(i+1+r) $.
\end{claim}
\begin{proof}
	Let $ \beta \in \set{0,1} $ be the value that $ \sigma(i-r) = \beta $ and $ \sigma(i+1+r) = \bar{\beta} $.
	Since $ |\Gamma_r(i) \cap \Gamma_r(i+1)| = |[i-r+1, i+r]| = 2r $,
	\begin{align*}
		\#_{\beta}(\sigma[\Gamma_r(i) \cap \Gamma_r(i+1)]) = r .
	\end{align*}
	Since $ \Gamma_r(i) = (\Gamma_r(i) \cap \Gamma_r(i+1)) \cup \set{i-r} $, it must be the case that $ \#_\beta(\Gamma_r(i)) = r+1 $.
	Similarly, since $ \Gamma_r(i+1) = (\Gamma_r(i) \cap \Gamma_r(i+1)) \cup \set{i+1+r} $, it must be the case that $ \#_{\bar{\beta}}(\Gamma_r(i+1)) = r+1 $.
	Hence, by the definition of $ \maj_r $,
	$ \sigma'(i) = \beta, \sigma'(i+1) = \bar{\beta} $ and the pair $ i,i+1 $ constitutes a \nameref{def:switch_point} in $ \sigma' $.
\end{proof}

\section{Temporally periodic configurations}\label{sec:temporally_periodic_configurations}

\begin{definition}\label{def:set_of_homogeneous_blocks}
	Given a configuration $ \sigma $ and a value $ \beta \in \set{0,1} $, we denote by $ B^\beta(\sigma) $ the set of maximal homogeneous blocks with value $ \beta $ in $ \sigma $.
	That is, 
	\[ 	B^\beta(\sigma) = \set{[i,j] \;:\; \forall k \in [i,j], \sigma(k)=\beta, \sigma(i-1)=\sigma(j+1) = \bar{\beta}} .\]
	Also, let $ B(\sigma) = B^0(\sigma) \cup B^1(\sigma) $.
\end{definition}

When we refer to an interval $ [i,j] \subseteq \cycnums{n} $, we denote its length by $ |[i,j]| $.
Note that it is not necessarily the case that $ |[i,j]| = j-i+1 $ because of the cyclical boundary conditions.

In \cite{goles1981_period_2}, it has been shown that a more general class of cellular automata that includes $ \maj_r $, always reach a cycle of temporal period 1 or 2.
Nevertheless, we provide a proof tailored for our special case, $ \maj_r $, because it is simpler and shorter than the general proof in \cite{goles1981_period_2}.

\begin{claim}\label{claim:period_2}
	For every integer $ r \ge 1 $, the rule $ \maj_r $ has temporal period $ 2 $.
\end{claim}
\begin{proof}
	Let $ \sigma_0 : \cycnums{n} \to \bitset $ be any initial configuration, and for any integer $ t \ge 0 $, let $ \sigma_t = \maj_r^t(\sigma_0) $.
	We define a potential function $ \phi : \set{0,1}^n \times \set{0,1}^n \rightarrow \mathbb{Z} $ over pairs of consecutive configurations in the sequence $ \set{\sigma_t}_{t=0}^{\infty} $.
	
	We shall use the following shorthand (where $i+x$ is computed mod $n$): 
	\begin{align}\label{eq:gt-def}
		g_t(i) = \sum_{j=i-r}^{i+r} \sigma_t(j) \;,
	\end{align}
	where we observe that if $g_t(i) \geq r+1$, then $\sigma_{t+1}(i) = 1$, while if $g_t(i) < r+1$, then $\sigma_{t+1}(i) = 0$.
	
	The potential function is defined as follows.
	
	\begin{align}\label{eq:varphi-def}
		\phi(\sigma_t, \sigma_{t-1}) = \sum_{i=0}^{n-1} \parentheses{\sigma_t(i) \cdot g_{t-1}(i)} - (r+1/2) \sum_{i=0}^{n-1} \parentheses{\sigma_t(i) + \sigma_{t-1}(i)}\;.
	\end{align}
	Observe that 
	\begin{align}\label{eq:symmetry}
		\sum_{i=0}^{n-1} \parentheses{\sigma_t(i) \cdot g_{t-1}(i)} = \sum_{i=0}^{n-1} \parentheses{\sigma_{t-1}(i) \cdot g_{t}(i)} \;.
	\end{align}
	(It is easiest to see this if we think of pairs $\sigma_t(i)=1$ and $\sigma_{t-1}(j)=1$ such that $|i-j| \leq r$ as edges.)
	
	Now consider the change in the potential function:
	\begin{eqnarray}
		\lefteqn{\phi(\sigma_{t+1}, \sigma_{t}) - \phi(\sigma_{t}, \sigma_{t-1})}\nonumber \\
		&=& \sum_{i=0}^{n-1} \parentheses{\sigma_{t+1}(i) \cdot g_{t}(i)}
		- (r+1/2) \sum_{i=0}^{n-1} \parentheses{\sigma_{t+1}(i) + \sigma_{t}(i)} \nonumber \\
		&& \;- \; \left(\sum_{i=0}^{n-1} \parentheses{\sigma_t(i) \cdot g_{t-1}(i)} - (r+1/2) \sum_{i=0}^{n-1} \parentheses{\sigma_t(i) + \sigma_{t-1}(i)} \right) \\
		&=& \sum_{i=0}^{n-1} \left( \parentheses{\sigma_{t+1}(i) \cdot g_{t}(i)} - \parentheses{\sigma_t(i) \cdot g_{t-1}(i)} \right)
		- (r+1/2) \sum_{i=0}^{n-1} \parentheses{\sigma_{t+1}(i) - \sigma_{t-1}(i)} \\
		&=&  \sum_{i=0}^{n-1} \parentheses{\sigma_{t+1}(i) - \sigma_{t-1}(i)} \cdot g_{t}(i)
		- (r+1/2) \sum_{i=0}^{n-1} \parentheses{\sigma_{t+1}(i) - \sigma_{t-1}(i)} \\
		&=&  \sum_{i=0}^{n-1} \parentheses{\sigma_{t+1}(i) - \sigma_{t-1}(i)} \cdot \parentheses{g_{t}(i) - (r+1/2) } \\
		&=&  \sum_{i: \sigma_{t+1}(i) \neq \sigma_{t-1}(i)} |g_{t}(i) - (r+1/2)| \\
		&\geq & \frac{1}{2} \left|\{i:  \sigma_{t+1}(i) \neq \sigma_{t-1}(i)\}\right|\;.
	\end{eqnarray}
	
	That is, the value of $ \phi $ is increasing at each step by $ \frac{1}{2} \left|\{i:  \sigma_{t+1}(i) \neq \sigma_{t-1}(i)\}\right|\ $, and since $ \phi $ is by definition a bounded function, it must be the case that there exists an integer $ t^* $ s.t. for every $ t > t^* $, $ \{i:  \sigma_{t+1}(i) \neq \sigma_{t-1}(i)\} = \emptyset $.
	Hence, $ \maj_r $ has temporal period $ 2 $.
\end{proof}

\begin{observation}\label{claim:final_locations}
	Let $ \sigma $ be a configuration and let $ [i,j] $ be an interval of \ref{def:cell}s such that for every \ref{def:cell} $ \ell \in [i,j] $, $ \sigma(\ell) = \beta $ for some $ \beta \in \set{0,1} $.
	For an integer $ t \ge 0 $, let $ \sigma_t = \maj_r^t(\sigma) $.
	If the length of the interval $ [i,j] $ is at least $ r+1 $, then for every $ t \ge 0 $ and for every \ref{def:cell} $ \ell \in [i,j] $, $ \sigma_t(\ell) = \beta $.
\end{observation}

Recall from \Cref{sec:result} that, given a configuration $ \sigma : \cycnums{n} \to \bitset $, we say that a \nameref{def:cell_interval} $ [i,j] $ is \nameref{def:strongly_stable}, \nameref{def:weakly_stable} or \nameref{def:unstable} if all the \ref{def:cell}s in that \nameref{def:cell_interval} are, respectively, \nameref{def:strongly_stable}, \nameref{def:weakly_stable} or \nameref{def:unstable}.
The definition applies to complete configurations as well.
In particular, we make a distinction that partitions the \nameref{def:temporally_periodic} configurations into two classes: \nameref{def:strongly_stable} configurations (configurations of the form $ (0^{r+1}0^* + 1^{r+1}1^*)^* $) and \nameref{def:weakly_stable} configurations (\nameref{def:temporally_periodic} configurations that are not \nameref{def:strongly_stable}). 
We also define \nameref{def:unstable} configurations as \nameref{def:transient} configurations.

\begin{observation}\label{claim:strongly_stable_configurations_are_fixed_points}
	If a configuration is \nameref{def:strongly_stable}, then it is also a \nameref{def:fixed_point}.
\end{observation}

\begin{observation}\label{claim:configurations_with_strongly_stable_locations_converge_to_strongly_stable_configurations}
	Let $ \sigma $ be a configuration.
	If $ B(\sigma) $ contains a block of length at least $ r+1 $, then there exists an integer $ t \ge 0 $ such that $ \maj_r^t(\sigma) $ is \nameref{def:strongly_stable}.
\end{observation}

\begin{claim}\label{claim:temporally_periodic_blocks_are_short}
	If $ \sigma $ is a \nameref{def:temporally_periodic} configuration, then it is either the case that for every block $ [i,j] \in B(\sigma) $, $ |[i,j]| \le r $ or that for every block $ [i,j] \in B(\sigma) $, $ |[i,j]| > r $.
\end{claim}
\begin{proof}
	Let $ \sigma' = \maj_r(\sigma) $ and $ \sigma'' = \maj_r(\sigma') $.
	Suppose by way of contradiction that $ B(\sigma) $ contains both a block of length at most $ r $ and a block of length at least $ r+1 $.
	Then there must be such a pair of consecutive blocks.
	Let $ [a,b] $ and $ [c,d] $ be the these two consecutive blocks where $ [a,b] $ is the block whose length is at most $ r $ and $ [c,d] $ is the block whose length is at least $ r+1 $.
	Since the two blocks are adjacent, it is either the case that $ c=b+1 $ or that $ a=d+1 $.
	Suppose without loss of generality that $ c=b+1 $.
	Let $ \beta $ be the value of the block $ [a,b] $.
	Since the block $ [c,d] $ is adjacent to the block $ [a,b] $, the value of $ [c,d] $ must be $ \bar{\beta} $.
	
	Since the \ref{def:cell} $ b $ belongs to $ [a,b] $, it must be the case that $ \sigma(b) = \beta $.
	However, for each \ref{def:cell} $ i \in [b+1,b+r] $, since $ |[c,d]| > r $, it must be the case that $ \sigma(i)=\bar{\beta} $.
	Also, since $ |[a,b]| \le r $ and $ a-1 \in [b-r,b] $, it must be the case that $ a-1 \in \Gamma_r(b) $, so $ \set{a-1} \cup [b+1,b+r] \in \Gamma_r(b) $.
	As $ \sigma(a-1) = \bar{\beta} $, it is the case that $ \#_{\bar{\beta}}(\sigma(\Gamma_r(b))) \ge r+1 $.
	Thus, $ \sigma'(b) = \bar{\beta} $.
	
	Since $ |[c,d]| \ge r+1 $ and for every \ref{def:cell} $ \ell \in [c,d] $, $ \sigma(\ell) = \bar{\beta} $, by \Cref{claim:final_locations}, it is also the case that for every \ref{def:cell} $ \ell \in [c,d] $, $ \sigma'(\ell) = \beta $ as well.
	Now, since $ |[b,d]| \ge r+1 $ as well, again, by \Cref{claim:final_locations}, it must hold that for every \ref{def:cell} $ \ell \in [b,d] $, $ \sigma''(\ell) = \bar{\beta} $ too.
	In particular,  $ \sigma''(b) = \bar{\beta} \ne \sigma(b) $, and so $ \sigma $ cannot be \nameref{def:temporally_periodic} and we reach a contradiction.
\end{proof}

\begin{corollary}\label{claim:weakly_stable_configurations_have_no_strongly_stable_cells}
	If a configuration $ \sigma $ is \nameref{def:weakly_stable}, then for each block $ [i,j] \in B(\sigma) $, $ |[i,j]| \le r $.
\end{corollary}

\begin{definition}[temporally periodic configuration pair]\label{def:temporally_periodic_pair}
	We say that a pair of configurations $ \sigma, \sigma' $ is a \nameref{def:temporally_periodic_pair} if $ \maj_r(\sigma) = \sigma' $ and $ \maj_r(\sigma') = \maj_r(\sigma) $.
\end{definition}

\begin{claim}\label{claim:single_switch_point_at_edges_of_temporally_periodic_block}
	Let $ \sigma, \sigma' $ be a \nameref{def:temporally_periodic_pair}.
	For every $ [i,j] \in B(\sigma') $, each of the intervals $ [i-r-1,j-r] $ and $ [i+r,j+r+1] $ contains exactly one \nameref{def:switch_point} in $ \sigma $.
\end{claim}
\begin{proof}
	We prove the claim for the interval $ [i-r-1,j-r] $, as the proof for the interval $ [i+r,j+r+1] $ is analogous.
	Let $ \beta \in \set{0,1} $ be the value such that $ [i,j] \in B^{\beta}(\sigma') $.
	Since $ [i,j] \in B^{\beta}(\sigma') $, each of the pairs $ (i-1,i) $ and $ (j,j+1) $ is a \nameref{def:switch_point} in $ \sigma' $.
	Since the pair $ \sigma,\sigma' $ is a \nameref{def:temporally_periodic_pair}, $ \maj_r(\sigma') = \sigma $.
	Hence, by \nameref{claim:switch_point_arg}, it must be the case that $ \sigma(i-r-1) = \bar{\beta} $ and $ \sigma(j-r) = \beta $.
	
	Therefore, there must be a \ref{def:cell} $ \ell \in [i-r-1,j-r] $ such that $ \sigma(\ell) = \bar{\beta} $ and $ \sigma(\ell+1) = \beta $.
	That is, the pair $ \ell, \ell+1 $ is a \nameref{def:switch_point} in $ \sigma $ in the interval $ [i-r-1,j-r] $.
	
	It is left to show that the pair $ \ell, \ell+1 $ is the only \nameref{def:switch_point} in the interval $ [i-r-1,j-r] $ in $ \sigma $.
	Suppose by way of contradiction that this is not the case.
	Let $ \ell',\ell'+1 $ be the \nameref{def:switch_point} in $ \sigma $ in $ [i-r-1,j-r] $ that is closest to $ \ell $.
	Since $ \ell,\ell+1 $ is a \nameref{def:switch_point} of values $ \bar{\beta}, \beta $, it must be the case that $ \ell',\ell'+1 $ is a \nameref{def:switch_point} of values $ \beta, \bar{\beta} $.
	That is, $ \sigma(\ell'+1) = \bar{\beta} $.
	Since the pair $ \sigma,\sigma' $ is a \nameref{def:temporally_periodic_pair}, $ \maj_r(\sigma) = \sigma' $.
	Hence, by \nameref{claim:switch_point_arg}, it must be the case that $ \sigma'(\ell'+1+r) = \bar{\beta} $.
	However, since $ \ell' \in [i-r-1,j-r] $, it follows that $ \ell'+1 \in [i-r,j-r] $, so $ \ell'+1+r \in [i,j] $.
	That is, the conclusion that $ \sigma'(\ell'+1+r) = \bar{\beta} $ is in contradiction to the assumption that $ [i,j] \in B^\beta(\sigma') $.
\end{proof}

\section{On the lengths of maximal homogeneous blocks}

\begin{definition}[balanced]\label{def:balanced}
	We say that a configuration $ \sigma $ is \nameref{def:balanced} if $ \#_0(\sigma) = \#_1(\sigma) $.
	Similarly, we say that an interval $ [i,j] $ is \nameref{def:balanced} in the configuration $ \sigma $ if $ \#_0(\sigma[i,j]) = \#_1(\sigma[i,j]) $.
\end{definition}

\begin{definition}[bias]\label{def:bias}
	Given a configuration $ \sigma $, we define its \nameref{def:bias} as $ \#_0(\sigma) - \#_1(\sigma) $.
	Similarly, given an interval $ [i,j] $, we define its \nameref{def:bias} in $ \sigma $ as $ \#_0(\sigma[i,j]) - \#_1(\sigma[i,j]) $.
\end{definition}

That is, a \nameref{def:balanced} configuration is a configuration whose \nameref{def:bias} is 0.
Similarly, an interval is \nameref{def:balanced} in a configuration if its \nameref{def:bias} is 0 in that configuration.

\begin{claim}\label{claim:block_length}
	Let $ \sigma $ and $ \sigma' $ be a pair of configurations where $ \maj_r(\sigma) = \sigma' $ and let $ [i,j] \in B^{\beta}(\sigma') $ for some $ \beta \in \set{0,1} $.
	If $ |[i,j]| \le 2r+1 $, then:
	\begin{align*}
		|[i,j]| = \#_\beta(\sigma[i-r,j+r]) - \#_{\bar{\beta}}(\sigma[j-r,i+r]) .
	\end{align*}
\end{claim}

\begin{proof}
	By the definition of $ B^{\beta} $, the interval $ [i,j] $ is a \ref*{def:block} in $ \sigma' $.
	Also, $ |[i,j]| \le 2r+1 < n $.
	Hence, the pair $ (i-1, i) $ is a \nameref{def:switch_point} in $ \sigma' $.
	Therefore, by \Cref{claim:switch-point-induces-balance}, the interval $ [i-r,i+r-1] $ is \nameref{def:balanced} in $ \sigma $.
	Similarly, since the pair $ (j,j+1) $ is also a \nameref{def:switch_point} in $ \sigma' $, again, by \Cref{claim:switch-point-induces-balance}, the interval $ [j-r+1,j+r] $ is \nameref{def:balanced} in $ \sigma $.
	
	We claim that the intervals $ [i-r,j-r] $ and $ [i+r,j+r] $ have the same \nameref{def:bias} in $ \sigma $.
	To see why, first observe that
	
	\begin{align}
		[i-r,j-r] = [i-r,i+r-1] \setminus [j-r+1, i+r-1] , \label{equation:[i-r,j-r]_length} \\
		[i+r,j+r] = [j-r+1,j+r] \setminus [j-r+1, i+r-1]. \label{equation:[i+r,j+r]_length} .
	\end{align}
	Since $ |[i-r,i+r-1]| = |[j-r+1,j+r]| = 2r $ and both $ [i-r,i+r-1] $ and $ [j-r+1,j+r] $ are \nameref{def:balanced} in $ \sigma $, by equations (\ref{equation:[i-r,j-r]_length}) and (\ref{equation:[i+r,j+r]_length}), the intervals $ [i-r,j-r] $ and $ [i+r,j+r] $ must have the same \nameref{def:bias} in $ \sigma $.
	The reason is that each of $ [i-r,j-r] $ and $ [i+r,j+r] $ equal the difference between a \nameref{def:balanced} interval of length $ 2r $ and the common sub-interval $ [j-r+1, i+r-1] $.
	
	Hence, since $ |[i-r,j-r]| = |[i+r,j+r]| $,
	
	\begin{align}
		&\#_0(\sigma[i+r,j+r]) = \#_0(\sigma[i-r,j-r]) \text{ and} \label{equation:same_num_of_zeros_in_r_shifted_intervals} \\
		&\#_1(\sigma[i+r,j+r]) = \#_1(\sigma[i-r,j-r]) .\label{equation:same_num_of_ones_in_r_shifted_intervals}
	\end{align}
	Since $ [j-r+1,j+r] = [j-r+1,i+r-1] \cup [i+r,j+r] $ and $ [j-r+1,j+r] $ is \nameref{def:balanced}, it must be the case that
	
	\begin{align*}
		&\#_0(\sigma[j-r+1,i+r-1]) + \#_0(\sigma[i+r,j+r]) = \\
		&\#_1(\sigma[j-r+1,i+r-1]) + \#_1(\sigma[i+r,j+r]) .
	\end{align*}
	Hence,
	
	\begin{align}
		\begin{split}\label{equation:num_zeros_in_[i+r,j+r]}
			\#_{\bar{\beta}}&(\sigma[i+r,j+r]) = \\
			&\#_{\beta}(\sigma[j-r+1,i+r-1]) + \#_{\beta}(\sigma[i+r,j+r]) - \#_{\bar{\beta}}(\sigma[j-r+1,i+r-1]) .
		\end{split}
	\end{align}
	We now express the length of $ [i,j] $.
	
	\begin{align}
		|[i,j]|
		&= |[i-r,j-r]| \\
		&= \#_0(\sigma[i-r,j-r]) + \#_1(\sigma[i-r,j-r]) \\
		&= \#_{\beta}(\sigma[i-r,j-r]) + \#_{\bar{\beta}}(\sigma[i+r,j+r]) \label{equation:[i,j]_len_equals_left_zeros_plus_right_ones} \\
		\begin{split}
			&= \#_{\beta}(\sigma[i-r,j-r]) \\
			&\;\;\;\; + \#_{\beta}(\sigma[j-r+1,i+r-1]) \\
			&\;\;\;\; + \#_{\beta}(\sigma[i+r,j+r]) \\
			&\;\;\;\; - \#_{\bar{\beta}}(\sigma[j-r+1,i+r-1])
		\end{split} \label{equation:block_length_compartmentalized} \\
		&= \#_\beta(\sigma[i-r,j+r]) - \#_{\bar{\beta}}(\sigma[j-r+1,i+r-1]) \label{equation:block_length} \\
		&= \#_\beta(\sigma[i-r,j+r]) - \#_{\bar{\beta}}(\sigma[j-r,i+r]) , \label{equation:block_length_as_claimed}
	\end{align}
	as claimed, where:
	\begin{itemize}
		\item[--] (\ref{equation:[i,j]_len_equals_left_zeros_plus_right_ones}) follows from applying equations (\ref{equation:same_num_of_zeros_in_r_shifted_intervals}) and (\ref{equation:same_num_of_ones_in_r_shifted_intervals}).
		
		\item[--] (\ref{equation:block_length_compartmentalized}) follows from applying Equation (\ref{equation:num_zeros_in_[i+r,j+r]}).
		
		\item[--] (\ref{equation:block_length}) follows from the fact that
		\begin{align*}
			[i-r,j-r] \cup [j-r+1,i+r-1] \cup [i+r,j+r] = [i-r,j+r] .
		\end{align*}
		\item[--] (\ref{equation:block_length_as_claimed}) follows from the observation that, by \nameref{claim:switch_point_arg}, since the pair $ (i,i-1) $ is a \nameref{def:switch_point} in $ \sigma' $, $ \sigma(i+r) = \sigma'(i) = \beta $, and, similarly, since the pair $ (j,j+1) $ is a \nameref{def:switch_point} in $ \sigma' $, $ \sigma(j-r) = \sigma'(j) = \beta $ as well.
		Hence, $ \#_{\bar{\beta}}(\sigma[j-r+1,i+r-1]) = \#_{\bar{\beta}}(\sigma[j-r,i+r]) $.
	\end{itemize}
We have thus established \Cref{claim:block_length}.
\end{proof}

\section{Block intervals defined by the left and right mappings}\label{sec:left_right_mappings}
\begin{definition}\label{def:block_interval}
	Given a configuration $ \sigma $, let $ B $ be a sequence of maximal homogeneous blocks in $ \sigma $.
	For two maximal homogeneous blocks $ X=[x,y] $ and $ X'=[x',y'] $, not necessarily belonging to $ B $, the \textsf{block-interval defined by the pair $ X,X' $}, denoted $ [X,X']_B $, is the following sequence of maximal homogeneous blocks:
	$$ [X,X']_B = \set{[i,j] \in B : [i,j] \subseteq [x,y']} . $$
\end{definition}

\begin{definition}\label{def:left_right_mapping}
	Let $ \sigma $ and $ \sigma' $ be a pair of configurations where $ \maj_r(\sigma) = \sigma' $.
	Given a maximal homogeneous block $ [i',j'] \in B(\sigma') $, let $ f^{\leftarrow}_{\sigma,\sigma'}([i',j']) $ be the maximal homogeneous block in $ B(\sigma) $ that contains the \ref{def:cell} $ i=j'-r $.
	Similarly, let $ f^{\rightarrow}_{\sigma,\sigma'}([i',j']) $ be the maximal homogeneous block in $ B(\sigma) $ that contains the \ref{def:cell} $ j=i'+r $.
	
	We refer to the function $ f^{\leftarrow}_{\sigma,\sigma'} $ as \textsf{the left mapping} from $ \sigma' $ to $ \sigma $, and, similarly, we refer to the function $ f^{\rightarrow}_{\sigma,\sigma'} $ as \textsf{the right mapping} from $ \sigma' $ to $ \sigma $.
\end{definition}

\begin{claim}\label{claim:mapping_preserves_values}
	For every pair of configurations $ \sigma $ and $ \sigma' $ satisfying $ \maj_r(\sigma) = \sigma' $, if $ [i',j'] \in B(\sigma') $, then the value of $ \sigma' $ at $ [i',j'] $ equals the value of $ \sigma $ at $ f^{\rightarrow}_{\sigma,\sigma'}([i',j']) $ as well as at $ f^{\leftarrow}_{\sigma,\sigma'}([i',j']) $.
\end{claim}
\begin{proof}
	By the definition of $ B(\sigma') $, the interval $ [i',j'] $ is a maximal homogeneous block.
	Therefore, the pair $ (i'-1,i') $ is a \nameref{def:switch_point} in $ \sigma' $, so by \nameref{claim:switch_point_arg}, $ \sigma(i'+r) = \sigma'(i') $.
	A similar argument holds for $ \sigma(j'-r) $, and the claim follows.
\end{proof}

\begin{claim}\label{claim:wrapping_block_interval_size_is_odd}
	For every pair of configurations $ \sigma $ and $ \sigma' $ satisfying $ \maj_r(\sigma) = \sigma' $ and a maximal homogeneous block $ [i',j'] \in B(\sigma') $, the number of homogeneous blocks in the block interval $ [f^{\leftarrow}_{\sigma,\sigma'}([i',j']), f^{\rightarrow}_{\sigma,\sigma'}([i',j'])]_{B(\sigma)} $ is odd.
\end{claim}
\begin{proof}
	By \Cref{claim:mapping_preserves_values}, the value of $ f^{\leftarrow}_{\sigma,\sigma'}([i,j]) $ equals the value of $ f^{\rightarrow}_{\sigma,\sigma'}([i',j']) $.
	Hence, since the values of \ref*{def:block}s in a configuration alternate, the number of \ref*{def:block}s in the block interval $ [f^{\leftarrow}_{\sigma,\sigma'}([i',j']), f^{\rightarrow}_{\sigma,\sigma'}([i',j'])]_{B(\sigma)} $ must be odd.
\end{proof}

\begin{claim}\label{claim:1to1}
	For every pair of configurations $ \sigma $ and $ \sigma' $ satisfying $ \maj_r(\sigma) = \sigma' $, each of the functions $ f^{\rightarrow}_{\sigma,\sigma'}([i',j']) $ and $ f^{\leftarrow}_{\sigma,\sigma'}([i',j']) $ is one-to-one.
\end{claim}
\begin{proof}
	We prove the claim for $ f^{\rightarrow}_{\sigma,\sigma'}([i',j']) $ and denote it by $ f $ for short.
	Let $ [i,j] \in B(\sigma) $ and suppose by way of contradiction that there are two different blocks $ [i',j'], [i'',j''] \in B(\sigma') $ s.t. $ f([i',j']) = f([i'',j'']) = [i,j] $.
	Let $ \beta $ be the value of $ \sigma $ at the block $ [i,j] $. 
	By \Cref{claim:mapping_preserves_values}, the value of $ \sigma' $ at both $ [i',j'] $ and $ [i'',j''] $ is also $ \beta $.
	By the definition of the mapping, the fact that $ f([i',j']) = f([i'',j'']) = [i,j] $ implies that both $ i' + r $ and $ i'' + r $ belong to $ [i,j] $.
	Hence, either $ [i' + r, i'' + r] \subseteq [i,j] $ or $ [i'' + r, i' + r] \subseteq [i,j] $.
	Assume without loss of generality that $ [i' + r, i'' + r] \subseteq [i,j] $.
	Hence, there must exist a maximal homogeneous block $ [i^*,j^*] $ in $ \sigma' $ where the value of $ \sigma' $ at $ [i^*,j^*] $ is $ \bar{\beta} $ and $ i^* \in [i',i''] $.
	Hence, it must hold that $ i^* + r \in [i' + r, i'' + r] \subseteq [i,j] $.
	That is, $ f([i^*,j^*]) = [i,j] $, in contradiction to \Cref{claim:mapping_preserves_values}.
\end{proof}

\begin{claim}\label{claim:number_of_blocks_is_non_decreasing}
	For every pair of configurations $ \sigma $ and $ \sigma' $ satisfying $ \maj_r(\sigma) = \sigma' $, it must be the case that $ |B(\sigma')| \le |B(\sigma)| $.
\end{claim}
\begin{proof}
	By \Cref{claim:1to1}, there is a one-to-one mapping from $ B(\sigma') $ to $ B(\sigma) $.
	Hence, $ |B(\sigma')| \le |B(\sigma)| $.
\end{proof}

\begin{claim}\label{claim:temporally_periodic_intervals_have_constant_lengths}
	If $ \sigma $ and $ \sigma' $ constitute a \nameref{def:temporally_periodic_pair}, then for every pair of blocks $ [a,b] \in B(\sigma) $ and $ [c,d] \in B(\sigma) $, the number of blocks in the block interval $ [f^{\leftarrow}_{\sigma',\sigma}([a,b]), f^{\rightarrow}_{\sigma',\sigma}([a,b])]_{B(\sigma')} $ equals the number of blocks in $ [f^{\leftarrow}_{\sigma',\sigma}([c,d]), f^{\rightarrow}_{\sigma',\sigma}([c,d])]_{B(\sigma')} $.
	
	Moreover, for every pair of blocks $ [a,b] \in B(\sigma) $ and $ [a',b'] \in B(\sigma') $, the number of blocks in the block interval $ [f^{\leftarrow}_{\sigma,\sigma'}([a',b']), f^{\rightarrow}_{\sigma,\sigma'}([a',b'])]_{B(\sigma)} $ equals the number of blocks in $ [f^{\leftarrow}_{\sigma',\sigma}([a,b]), f^{\rightarrow}_{\sigma',\sigma}([a,b])]_{B(\sigma')} $.
\end{claim}
\begin{proof}
	Let $ k $ be the number of maximal homogeneous blocks in $ \sigma $.
	By \Cref{claim:1to1}, the functions $ f^{\leftarrow}_{\sigma,\sigma'} $ and $ f^{\rightarrow}_{\sigma,\sigma'} $ are one-to-one, and therefore the number of maximal homogeneous blocks in $ \sigma' $ is $ k $ as well.
	Denote by $ [a_1,b_1], \dots, [a_k,b_k] $ the sequence of maximal homogeneous blocks in $ \sigma $ starting from an arbitrary block $ [a_1,b_1] $ such that for each $ 1 \le i \le k-1 $, $ a_{i+1} = b_i + 1 $ (that is, the blocks $ [a_i,b_i] $ and $ [a_{i+1},b_{i+1}] $ are consecutive).
	
	We claim that for every $ 1 \le i \le k $, if $ f^{\leftarrow}_{\sigma',\sigma}([a_i,b_i]) = [a'_i,b'_i] $ and $ f^{\leftarrow}_{\sigma',\sigma}([a_{i+1},b_{i+1}]) = [a'_{i+1},b'_{i+1}] $, then $ a'_{i+1} = b'_i + 1 $ (that is, the blocks $ f^{\leftarrow}_{\sigma',\sigma}([a_i,b_i]) $ and $ f^{\leftarrow}_{\sigma',\sigma}([a_{i+1},b_{i+1}]) $ are consecutive).
	We also claim that for every $ 1 \le i \le k $, if $ f^{\rightarrow}_{\sigma',\sigma}([a_i,b_i]) = [a''_i,b''_i] $ and $ f^{\rightarrow}_{\sigma',\sigma}([a_{i+1},b_{i+1}]) = [a''_{i+1},b''_{i+1}] $, then $ a''_{i+1} = b''_i + 1 $.
	We prove the former (since the proof of the latter is analogous).
	
	Suppose towards a contradiction that for some $ 1 \le i \le k-1 $, $ a'_{i+1} \ne b'_i + 1 $.
	In that case, there must exist an integer $ 2 \le j \le k-i $ for which
	
	\begin{align*}
		f^{\leftarrow}_{\sigma,\sigma'}([a_{i+j},b_{i+j}]) \subseteq [a'_i,b'_{i+1}] .
	\end{align*}
	Since $ b_{i+j}-r \in f^{\leftarrow}_{\sigma,\sigma'}([a_{i+j},b_{i+j}]) $, this means that $ b_{i+j}-r \in [a'_i,b'_{i+1}] $, and hence $ b_{i+j}-r \in [b_i-r,b_{i+1}-r] $.
	Thus,
	
	\begin{align*}
		b_{i+j} \in [b_i,b_{i+1}] \subseteq [a_i,b_i] \cup [a_{i+1},b_{i+1}] ,
	\end{align*}
	in contradiction to $ j \ge 2 $ (in other words, in contradiction to $ [a_{i+j},b_{i+j}] $ being distinct from $ [a_i,b_i] $ and $ [a_{i+1},b_{i+1}] $).
	
	This establishes the claim that $ a'_{i+1} = b'_i + 1 $, and an analogous argument implies that $ a''_{i+1} = b''_i + 1 $ as well.
	Hence,
	
	\begin{align*}
		[f^{\leftarrow}_{\sigma,\sigma'}&([a_{i+1},b_{i+1}]), f^{\rightarrow}_{\sigma,\sigma'}([a_{i+1},b_{i+1}])]_{B(\sigma)} \\
		&= [f^{\leftarrow}_{\sigma,\sigma'}([a_{i},b_{i}]), f^{\rightarrow}_{\sigma,\sigma'}([a_{i},b_{i}])]_{B(\sigma)} \cup \set{[a''_i,b''_i]} \setminus \set{[a'_i,b'_i]} .
	\end{align*}
	Thus,
	
	\begin{align}\label{equation:consecutive_block_intervals_are_equal}
		|[f^{\leftarrow}_{\sigma,\sigma'}([a_{i+1},b_{i+1}]), f^{\rightarrow}_{\sigma,\sigma'}([a_{i+1},b_{i+1}])]_{B(\sigma)}| = |[f^{\leftarrow}_{\sigma,\sigma'}([a_{i},b_{i}]), f^{\rightarrow}_{\sigma,\sigma'}([a_{i},b_{i}])]_{B(\sigma)}| .
	\end{align}
	Since Equation~(\ref{equation:consecutive_block_intervals_are_equal}) holds for every $ 1 \le i \le k $, it must be the case that for every pair of blocks $ [a,b],[c,d] \in B(\sigma) $, the number of blocks in the block interval $ [f^{\leftarrow}_{\sigma,\sigma'}([a,b]), f^{\rightarrow}_{\sigma,\sigma'}([a,b])]_{B(\sigma)} $ equals the number of blocks in the block interval $ [f^{\leftarrow}_{\sigma,\sigma'}([c,d]), f^{\rightarrow}_{\sigma,\sigma'}([c,d])]_{B(\sigma)} $, which establishes the first part of the claim.
\end{proof}

\section{The \nameref*{def:alignment_mapping}}\label{sec:alignment_mapping}

\begin{definition}[alignment mapping]\label{def:alignment_mapping}
	Let $ \sigma $ and $ \sigma' $ be a pair of configurations satisfying $ \maj_r(\sigma) = \sigma' $.
	Given a maximal homogeneous block $ [i,j] \in B(\sigma') $, let $ \varphi_{\sigma,\sigma'}([i,j]) $ be the middle block in the block interval $ [f^{\leftarrow}_{\sigma,\sigma'}([i,j]), f^{\rightarrow}_{\sigma,\sigma'}([i,j])]_{B(\sigma)} $ (the middle block is well-defined, since, by \Cref{claim:wrapping_block_interval_size_is_odd}, the number of blocks in that interval is odd).
	
	We refer to the function $ \varphi_{\sigma,\sigma'} $ as the \nameref{def:alignment_mapping} from $ \sigma' $ to $ \sigma $.	
\end{definition}

\begin{claim}\label{claim:alignment_mapping_is_one_to_one}
	For every pair of configurations $ \sigma $ and $ \sigma' $ satisfying $ \maj_r(\sigma) = \sigma' $, the \nameref{def:alignment_mapping} $ \varphi_{\sigma,\sigma'} $ is one-to-one.
\end{claim}
\begin{proof}
	Suppose that, contrary to the claim, $ \varphi_{\sigma,\sigma'} $ is not one-to-one.
	That is, there are two distinct blocks $ [a,b], [c,d] \in B(\sigma') $ such that
	
	\begin{align*}
		\varphi_{\sigma,\sigma'}([a,b]) = \varphi_{\sigma,\sigma'}([c,d]) .
	\end{align*}
	That being the case, denote by $ I \in B(\sigma) $ the block satisfying $ I = \varphi_{\sigma,\sigma'}([a,b]) = \varphi_{\sigma,\sigma'}([c,d]) $.
	
	Since $ [a,b] $ and $ [c,d] $ are maximal homogeneous blocks, it is either the case that the intervals $ [b,c] $ and $ [a,d] $ satisfy $ [b,c] \subseteq [a,d] $ or that they satisfy $ [a,d] \subseteq [b,c] $.
	We assume, then, without loss of generality, that
	
	\begin{align}\label{equation:[b,c]_in_[a,d]}
		[b,c] \subseteq [a,d] .
	\end{align}
	By \Cref{claim:wrapping_block_interval_size_is_odd}, for every maximal homogeneous block $ [i,j] \in B(\sigma') $, the number of blocks in $ [f^{\leftarrow}_{\sigma,\sigma'}([i,j]), f^{\rightarrow}_{\sigma,\sigma'}([i,j])]_{B(\sigma)} $ is odd.
	Let $ \delta_{[a,b]} $ be the integer satisfying
	
	\begin{align*}
		|[f^{\leftarrow}_{\sigma,\sigma'}([a,b]), f^{\rightarrow}_{\sigma,\sigma'}([a,b])]_{B(\sigma)}| = 2\delta_{[a,b]} + 1 ,
	\end{align*}
	and, similarly, let $ \delta_{[c,d]} $ be the integer satisfying
	
	\begin{align*}
		|[f^{\leftarrow}_{\sigma,\sigma'}([c,d]), f^{\rightarrow}_{\sigma,\sigma'}([c,d])]_{B(\sigma)}| = 2\delta_{[c,d]} + 1 .
	\end{align*}
	
	By \Cref{def:alignment_mapping}, the block $ I $ defined above is the middle block of the block interval $ [f^{\leftarrow}_{\sigma,\sigma'}([a,b]), f^{\rightarrow}_{\sigma,\sigma'}([a,b])]_{B(\sigma)} $ as well as of the block interval $ [f^{\leftarrow}_{\sigma,\sigma'}([c,d]), f^{\rightarrow}_{\sigma,\sigma'}([c,d])]_{B(\sigma)} $.
	This means that the block $ f^{\leftarrow}_{\sigma,\sigma'}([a,b]) $ is located $ \delta_{[a,b]} $ blocks away from $ I $ to its left side and that the block $ f^{\rightarrow}_{\sigma,\sigma'}([a,b]) $ is located $ \delta_{[a,b]} $ blocks away from $ I $ to its right side.
	Similarly, the block $ f^{\leftarrow}_{\sigma,\sigma'}([c,d]) $ is located $ \delta_{[c,d]} $ blocks away from $ I $ to the left, and that the block $ f^{\rightarrow}_{\sigma,\sigma'}([c,d]) $ is located $ \delta_{[c,d]} $ blocks away from $ I $ to the right.
	
	Thus, if $ \delta_{[a,b]} \le \delta_{[c,d]} $, then
	
	\begin{align}\label{equation:f([a,b]_in_[c,d])}
		[f^{\leftarrow}_{\sigma,\sigma'}([a,b]), f^{\rightarrow}_{\sigma,\sigma'}([a,b])]_{B(\sigma)}
		\subseteq 
		[f^{\leftarrow}_{\sigma,\sigma'}([c,d]), f^{\rightarrow}_{\sigma,\sigma'}([c,d])]_{B(\sigma)} ,
	\end{align}
	and if $ \delta_{[a,b]} \ge \delta_{[c,d]} $, then
	
	\begin{align}\label{equation:f([c,d]_in_[a,b])}
		[f^{\leftarrow}_{\sigma,\sigma'}([c,d]), f^{\rightarrow}_{\sigma,\sigma'}([c,d])]_{B(\sigma)}
		\subseteq 
		[f^{\leftarrow}_{\sigma,\sigma'}([a,b]), f^{\rightarrow}_{\sigma,\sigma'}([a,b])]_{B(\sigma)} .
	\end{align}
	In the former case (Equation~(\ref{equation:f([a,b]_in_[c,d])})),
	
	\begin{align}\label{equation:f([a,b]_in_[c,d])_implication}
		[b-r,a+r] \subseteq [d-r,c+r] .
	\end{align}
	This is because, by \Cref{def:left_right_mapping},
	\begin{align*}
		b-r &\in f^{\leftarrow}_{\sigma,\sigma'}([a,b]) , \\
		a+r &\in f^{\rightarrow}_{\sigma,\sigma'}([a,b]) , \\
		d-r &\in f^{\leftarrow}_{\sigma,\sigma'}([c,d]) , \text{ and} \\
		c+r &\in f^{\rightarrow}_{\sigma,\sigma'}([c,d]) .
	\end{align*}
	For the same reason, in the latter case (Equation~(\ref{equation:f([c,d]_in_[a,b])})),
	
	\begin{align}\label{equation:f([c,d]_in_[a,b])_implication}
		[d-r,c+r] \subseteq [b-r,a+r] .
	\end{align}
	However, as we show next, both $ [b-r,a+r] \subseteq [d-r,c+r] $ (Equation~(\ref{equation:f([a,b]_in_[c,d])_implication})) and $ [d-r,c+r] \subseteq [b-r,a+r] $ (Equation~(\ref{equation:f([c,d]_in_[a,b])_implication})) are impossible given the assumption that $ [b,c] \subseteq [a,d] $ (Equation~(\ref{equation:[b,c]_in_[a,d]})).
	
	To see why $ [b-r,a+r] \subseteq [d-r,c+r] $ contradicts $ [b,c] \subseteq [a,d] $, observe that $ [b,c] \subseteq [a,d] $ implies $ [d-r,c+r] \subseteq [d-r,d+r] $, and since $ [d-r,d+r] = 2r+1 $, it must be the case that $ |[d-r,c+r]| < 2r $.
	
	Additionally, $ [b,c] \subseteq [a,d] $ also implies $ [b-r,b+r] \subseteq [b-r,c+r] $ and $ |[b-r,b+r]| = 2r+1 $, so $ |[b-r,c+r]| > 2r $.
	
	However, $ [b-r,a+r] \subseteq [d-r,c+r] $ implies that $ [b-r,c+r] \subseteq [d-r,c+r] $, so it cannot be the case that both $ |[d-r,c+r]| < 2r $ and $ |[b-r,c+r]| > 2r $.
	
	The case in which $ [d-r,c+r] \subseteq [b-r,a+r] $ can similarly be shown to contradict $ [b,c] \subseteq [a,d] $, so we reach a contradiction in either case.
\end{proof}

\begin{claim}\label{claim:adjacent_blocks_in_temporally_periodic_pairs_are_mapped_into_adjacent_blocks}
	Let $ \sigma $, $ \sigma' $ be a \nameref{def:temporally_periodic_pair} and let $ [a,b], [c,d] \in B(\sigma') $ be two adjacent blocks in which $ c=b+1 $.
	If $ \varphi_{\sigma,\sigma'}([a,b]) = [a',b'] $ and $ \varphi_{\sigma,\sigma'}([c,d]) = [c',d'] $, then the blocks $ [a',b'] $ and $ [c',d'] $ are also adjacent and $ c'=b'+1 $.
\end{claim}
\begin{proof}
	Let $ \beta \in \set{0,1} $ be the value of the block $ [a,b] $ in $ \sigma' $.
	Since the block $ [c,d] $ is adjacent to the block $ [a,b] $, the value of the block $ [c,d] $ in $ \sigma' $ must be $ \bar{\beta} $.
	By \nameref{claim:switch_point_arg}, $ \sigma(a+r) = \beta $ and $ \sigma(c+r) = \bar{\beta} $, so the value of the block $ f^{\rightarrow}_{\sigma,\sigma'}([a,b]) $ in $ \sigma $ is $ \beta $ and the value of the block $ f^{\rightarrow}_{\sigma,\sigma'}([c,d]) $ in $ \sigma' $ is $ \bar{\beta} $.
	
	By \Cref{claim:single_switch_point_at_edges_of_temporally_periodic_block}, the interval $ [a+r,b+r+1] $ contains exactly one \nameref{def:switch_point} in $ \sigma $.
	Since $ c = b+r $, the interval $ [a+r,c+r] $ contains exactly one \nameref{def:switch_point} in $ \sigma $.
	As both $ f^{\rightarrow}_{\sigma,\sigma'}([a,b]) $ and $ f^{\rightarrow}_{\sigma,\sigma'}([c,d]) $ intersect with the interval $ [a+r,c+r] $, that single \nameref{def:switch_point} must be the $ \beta,\bar{\beta} $ \nameref{def:switch_point} between $ f^{\rightarrow}_{\sigma,\sigma'}([a,b]) $ and $ f^{\rightarrow}_{\sigma,\sigma'}([c,d]) $.
	Hence, the blocks $ f^{\rightarrow}_{\sigma,\sigma'}([a,b]) $ and $ f^{\rightarrow}_{\sigma,\sigma'}([c,d]) $ are adjacent with the block $ f^{\rightarrow}_{\sigma,\sigma'}([a,b]) $ preceding the block $ f^{\rightarrow}_{\sigma,\sigma'}([c,d]) $.
	
	By a similar argument, the blocks $ f^{\leftarrow}_{\sigma,\sigma'}([a,b]) $ and $ f^{\leftarrow}_{\sigma,\sigma'}([c,d]) $ are adjacent with the block $ f^{\leftarrow}_{\sigma,\sigma'}([a,b]) $ preceding the block $ f^{\leftarrow}_{\sigma,\sigma'}([c,d]) $.
	
	By the definition of the \nameref{def:alignment_mapping}, the block $ [a',b'] $ is the middle block of the block interval $ [f^{\leftarrow}_{\sigma,\sigma'}([a,b]), f^{\rightarrow}_{\sigma,\sigma'}([a,b])]_{B(\sigma)} $, and the block $ [c',d'] $ is the middle block of the block interval $ [f^{\leftarrow}_{\sigma,\sigma'}([c,d]), f^{\rightarrow}_{\sigma,\sigma'}([c,d])]_{B(\sigma)} $.
	
	Therefore, the blocks $ [a',b'] $ and $ [c',d'] $ must be adjacent to each other and it must also be the case that $ c'=b'+1 $, as claimed.
\end{proof}

\begin{definition}\label{def:iterative_application_of_alignment_mapping_in_spatially_periodic_pairs}
	Let $ \sigma $, $ \sigma' $ be a \nameref{def:temporally_periodic_pair}.
	Given an integer $ k $, we define $ \varphi_{\sigma,\sigma'}^k $ as follows.
	For every block $ [i,j] \in B(\sigma) $,
	\begin{enumerate}
		\item $ \varphi_{\sigma,\sigma'}^0([i,j]) = [i,j] $.
		\item $ \varphi_{\sigma,\sigma'}^1([i,j]) = \varphi_{\sigma,\sigma'}([i,j]) $.
		\item for $ k>1 $, if $ k $ is odd, then $ \varphi_{\sigma,\sigma'}^k([i,j]) = \varphi_{\sigma,\sigma'}(\varphi_{\sigma,\sigma'}^{k-1}([i,j])) $.
		\item for $ k>1 $, if $ k $ is even, then $ \varphi_{\sigma,\sigma'}^k([i,j]) = \varphi_{\sigma',\sigma}(\varphi_{\sigma,\sigma'}^{k-1}([i,j])) $.
	\end{enumerate}
\end{definition}

\begin{observation}\label{claim:alignment_mapping_preserves_value_in_weakly_stable_pairs}
	Let $ \sigma $, $ \sigma' $ be a \nameref{def:temporally_periodic_pair}.
	For every block $ [i,j] \in B(\sigma') $, $ \varphi_{\sigma,\sigma'}^2([i,j]) = [i,j] $.
\end{observation}

\begin{definition}\label{def:iterative_application_of_alignment_mapping_in_general}
	Let $ \sigma_0 : \cycnums{n} \to \bitset $ be any initial configuration, and for any integer $ t \ge 0 $, let $ \sigma_t = \maj_r^t(\sigma_0) $.
	Given a time step $ t \ge 1 $, we define $ \varphi_t $ as $ \varphi_{\sigma_{t-1},\sigma_t}([i,j]) $ for every block $ [i,j] \in B(\sigma_t) $.
	Given an integer $ k $ and a time step $ t $ s.t. $ t \ge k $, we define $ \varphi_{t}^k $ as follows.
	For every block $ [i,j] \in B(\sigma_t) $:
	\begin{enumerate}
		\item $ \varphi_{t}^0([i,j]) = [i,j] $.
		\item for $ k>1 $, $ \varphi_{t}^k([i,j]) = \varphi_{t-k+1}(\varphi_{t}^{k-1}([i,j])) $.
	\end{enumerate}
\end{definition}

\section{Block lengths in \nameref*{def:temporally_periodic} configurations}\label{sec:block_lengths}

\begin{claim}\label{claim:block_length_for_temporally_periodic_configurations}
	Let $ \sigma $, $ \sigma' $ be a \nameref{def:temporally_periodic_pair} of \nameref{def:weakly_stable} configurations and let $ \beta \in \set{0,1} $. 
	For every block $ [i,j] \in B^{\beta}(\sigma') $,
	\begin{align}\label{[i,j]_length_temporally_periodic}
		|[i,j]| = \sum_{[i',j']\in A^{\beta}} |[i',j']| - \sum_{[i',j']\in A^{\bar{\beta}}} |[i',j']| ,
	\end{align}
	where 
	$$ A^{\beta} = [f^{\leftarrow}_{\sigma,\sigma'}([i,j]), f^{\rightarrow}_{\sigma,\sigma'}([i,j])]_{B^{\beta}(\sigma)} , $$
	and, similarly,
	$$ A^{\bar{\beta}} = [f^{\leftarrow}_{\sigma,\sigma'}([i,j]), f^{\rightarrow}_{\sigma,\sigma'}([i,j])]_{B^{\bar{\beta}}(\sigma)} . $$
\end{claim}
\begin{proof}
	Since $ \sigma' $ is a \nameref{def:temporally_periodic} \nameref{def:weakly_stable} configuration, by \Cref{claim:temporally_periodic_blocks_are_short}, $ |[i,j]| \le r < 2r + 1 $. 
	Therefore, the conditions for applying \Cref{claim:block_length} hold.
	So, by \Cref{claim:block_length},
	
	\begin{align}\label{equation:[i,j]_length}
		|[i,j]| = \#_\beta(\sigma[i-r,j+r]) - \#_{\bar{\beta}}(\sigma[j-r,i+r]) .
	\end{align}
	
	To prove the claim, we relate Equation~(\ref{equation:[i,j]_length}) to Equation~(\ref{[i,j]_length_temporally_periodic}) by defining four sets of \ref{def:cell}s:
	
	\begin{align*}
		X^{\beta} &= \set{\ell \in [i',j'] \;:\; [i',j'] \in A^{\beta}}, \\
		X^{\bar{\beta}} &= \set{\ell \in [i',j'] \;:\; [i',j'] \in A^{\bar{\beta}}}, \\
		Y^{\beta} &= \set{\ell \in [i-r, j+r]| \sigma(\ell) = \beta}, \\
		Y^{\bar{\beta}} &= \set{\ell \in [j-r, i+r]| \sigma(\ell) = \bar{\beta}} .
	\end{align*}
	In order to prove the claim, it is sufficient to show that
	
	\begin{align*}
		X^{\beta} &= Y^{\beta}, \\
		X^{\bar{\beta}} &= Y^{\bar{\beta}}.
	\end{align*}
	
	We first show that $ X^{\bar{\beta}} = Y^{\bar{\beta}} $.
	
	Let $ i^* $ be the leftmost \ref{def:cell} in the block $ f^{\leftarrow}_{\sigma,\sigma'}([i,j]) $ and let $ j^* $ be the rightmost \ref{def:cell} in the block $ f^{\rightarrow}_{\sigma,\sigma'}([i,j]) $.
	By the definition of $ A^{\bar{\beta}} $, 
	
	\begin{align*}
		X^{\bar{\beta}} = \set{\ell \in [i^*,j^*] \;:\; \sigma(\ell) = \bar{\beta}}.
	\end{align*}
	Since $ [i,j] \in B^{\beta}(\sigma') $, each of the pairs $ (i-1,i) $ and $ (j,j+1) $ is a \nameref{def:switch_point} in $ \sigma' $, so by \nameref{claim:switch_point_arg}, $ \sigma(j-r) = \sigma(i+r) = \beta $.
	As $ f^{\leftarrow}_{\sigma,\sigma'}([i,j]) $ and $ f^{\rightarrow}_{\sigma,\sigma'}([i,j]) $ are by definition the two maximal homogeneous blocks in $ \sigma $ that contain the \ref{def:cell}s $ j-r $ and $ i+r $ respectively, it must be the case that for each \ref{def:cell} $ \ell \in [i^*,j-r] \cup [i+r,j^*] $, $ \sigma(\ell) = \beta $.
	Hence, if a \ref{def:cell} $ \ell \in [i^*,j^*] $ satisfies $ \sigma(\ell) = \bar{\beta} $, then $ \ell \in [j-r+1,i+r-1] \subseteq [j-r,i+r] $.
	That is,
	
	\begin{align*}
		X^{\bar{\beta}} = \set{\ell \in [j-r, i+r]| \sigma(\ell) = \bar{\beta}} = Y^{\bar{\beta}} .
	\end{align*}
	
	We now show that $ X^{\beta} = Y^{\beta} $.
	
	Recall that $ i^* $ is the leftmost \ref{def:cell} in the block $ f^{\leftarrow}_{\sigma,\sigma'}([i,j]) $ and that $ j^* $ is the rightmost \ref{def:cell} in the block $ f^{\rightarrow}_{\sigma,\sigma'}([i,j]) $, which means that
	$ X^{\beta} = \set{\ell \in [i^*,j^*] \;:\; \sigma(\ell) = \beta} $.
	
	Since $ (i-1,i) $ is a \nameref{def:switch_point} in $ \sigma' $, by \nameref{claim:switch_point_arg}, $ \sigma(i-r-1) = \bar{\beta} $.
	This implies that $ i^* \in [i-r,j-r] $.
	As $ (i^*-1,i^*) $ is a \nameref{def:switch_point} in $ \sigma $ of type $ (\bar{\beta}, \beta) $, by \Cref{claim:single_switch_point_at_edges_of_temporally_periodic_block}, there is exactly one \nameref{def:switch_point} in $ [i-r-1,j-r] $, and therefore it must be the case that
	
	\begin{align}\label{equation:shrink_left_interval}
		\set{\ell \in [i-r,j-r] \;:\; \sigma(\ell) = \beta} = [i^*,j-r] .
	\end{align}
	
	Similarly, since $ (j,j+1) $ is a \nameref{def:switch_point} in $ \sigma' $, by \nameref{claim:switch_point_arg}, $ \sigma(j+r+1) = \bar{\beta} $, which implies that $ j^* \in [i+r,j+r] $.
	Because $ (i^*-1,i^*) $ is a \nameref{def:switch_point} in $ \sigma $ of type $ (\bar{\beta}, \beta) $, by \Cref{claim:single_switch_point_at_edges_of_temporally_periodic_block}, there is exactly one \nameref{def:switch_point} in $ [i+r,j+r+1] $, and therefore it must be the case that
	
	\begin{align}
		\set{\ell \in [i+r,j+r] \;:\; \sigma(\ell) = \beta} = [i+r,j^*] .
	\end{align}
	Consequently,
	
	\begin{align}\label{equation:shrink_right_interval}
		Y^{\beta} 
		&= \set{\ell \in [i-r, j+r]| \sigma(\ell) = \beta} \\
		&= \set{\ell \in [i^*,j^*] \;:\; \sigma(\ell) = \beta} \label{equation:shrink_interval_both_sides} \\
		&= X^{\beta},
	\end{align}
	where Equation~(\ref{equation:shrink_interval_both_sides}) follows from equations (\ref{equation:shrink_left_interval}) and (\ref{equation:shrink_right_interval}) together with the observations that $ i^* \in [i-r,j-r] $ and $ j^* \in [i+r,j+r] $.
	
	We've shown that $ X^{\beta} = Y^{\beta} $ and $ X^{\bar{\beta}} = Y^{\bar{\beta}} $, so the claim follows.
\end{proof}

\section{The \nameref*{def:block_length_vector}s of \nameref*{def:temporally_periodic_pair}s}\label{sec:block_length_vectors}

\begin{definition}[block-length vector]\label{def:block_length_vector}
	Given a configuration $ \sigma : \cycnums{n} \to \bitset $, we define its \nameref{def:block_length_vector} $ \vec{v}: \cycnums{|B(\sigma)|} \to \mathbb{N} $ as the cyclic sequence of the lengths of the configuration's maximal homogeneous blocks.
	
	That is, $ \vec{v}(0) = |[i,j]| $ for an arbitrary block $ [i,j] \in B(\sigma) $, and for each $ k \in \cycnums{|B(\sigma)|} $, if $ [a,b] \in B(\sigma) $ is the block for which $ \vec{v}(k-1) = |[a,b]| $ and $ [c,d] \in B(\sigma) $ is the block satisfying $ c=b+1 $, then $ \vec{v}(k) = |[c,d]| $.
\end{definition}

We note that every possible \nameref{def:block_length_vector}, viewed as a ring of integers, corresponds to at most two configurations (up to cyclic shifts): one where the blocks at the odd positions have a value of 0 and one where the blocks at the even positions have a value of 0.\footnote{
	These two possible configurations collapse into one in the case in which the \nameref{def:block_length_vector} $ \vec{v} $ equals the concatenation of some \nameref{def:block_length_vector} $ \vec{u} $ to itself (i.e, $ \vec{v} = \vec{u}\vec{u} $), where $ \vec{u} $ is of odd length (for instance, the \nameref{def:block_length_vector} $ 123123 $ corresponds to exactly one configuration, up to a cyclic shift, as the configuration resulting from assigning the value 0 to the first block is the same configuration resulting from assigning the value 1 to the first block).
}
So when we say that a \nameref{def:block_length_vector} corresponds to a configuration $ \sigma $, it means that the configuration $ \sigma $ can be either of the at most two possibilities.

\begin{definition}\label{def:length_of_block_length_vector}
	We define the length of a \nameref{def:block_length_vector} $ \vec{v} $, denoted by $ |\vec{v}| $, as the number of entries in the vector.
	That is, if $ \vec{v} $ is the \nameref{def:block_length_vector} that corresponds to the configuration $ \sigma $, then $ |\vec{v}| = |B(\sigma)| $.
\end{definition}

\begin{claim}\label{claim:temporally_periodic_vectors_have_same_length}
	If $ \vec{v} $ and $ \vec{v'} $ are a pair of \nameref{def:block_length_vector}s corresponding to a \nameref{def:temporally_periodic_pair}, then $ |\vec{v}| = |\vec{v'}| $.
\end{claim}
\begin{proof}
	Let $ \sigma $ be a configuration that corresponds to the \nameref{def:block_length_vector} $ \vec{v} $ and let $ \sigma'=\maj_r(\sigma) $.
	Clearly, $ \vec{v'} $ is the the \nameref{def:block_length_vector} of the configuration $ \sigma' $, and the pair $ \sigma, \sigma' $ is a \nameref{def:temporally_periodic_pair}.
	
	Since $ \sigma'=\maj_r(\sigma) $, by \Cref{claim:number_of_blocks_is_non_decreasing}, $ |B(\sigma')| \le |B(\sigma)| $.
	Since $ \sigma, \sigma' $ are a \nameref{def:temporally_periodic_pair}, it is also the case that $ \sigma=\maj_r(\sigma') $, so again by \Cref{claim:number_of_blocks_is_non_decreasing}, $ |B(\sigma)| \le |B(\sigma')| $.
	That is, $ |B(\sigma)| = |B(\sigma')| $ and the claim follows.
\end{proof}

\begin{definition}[aligned]\label{def:aligned}
	Let $ \vec{v} $ and $ \vec{v'} $ be a pair of \nameref{def:block_length_vector}s of length $ k $ each corresponding to a \nameref{def:temporally_periodic_pair} $ \sigma $ and $ \sigma' $.
	For each $ i \in \cycnums{k} $, let $ I_i $ the \ref*{def:block} in $ \sigma $ that corresponds to $ \vec{v}_i $ and let $ I'_i $ be the \ref*{def:block} in $ \sigma' $ that corresponds to $ \vec{v'}_i $.
	We say that the pair $ \vec{v} $ and $ \vec{v'} $ are \nameref{def:aligned} if for each $ i \in \cycnums{k} $, $ \varphi_{\sigma,\sigma'}(I'_i) = I_i $ (where $ \varphi_{\sigma,\sigma'} $ is the \nameref{def:alignment_mapping} from $ \sigma' $ to $ \sigma $).
\end{definition}

\section{The \nameref*{def:horizon} of \nameref*{def:block_length_vector}s}\label{sec:horizon}

\begin{definition}[horizon]\label{def:horizon}
	Let $ \vec{v} $ and $ \vec{v'} $ be a pair of \nameref{def:block_length_vector}s of lengths $ k $ and $ k' $ each, corresponding to a pair of configurations $ \sigma $ and $ \sigma' $ satisfying $ \maj_r(\sigma) = \sigma' $.
	For each entry $ i \in \cycnums{k'} $ of $ \vec{v'} $, we define the \nameref{def:horizon} of $ i $ in $ \vec{v'} $, denoted by $ \delta_{\vec{v'}}(i) $, as follows:
	$ \delta_{\vec{v'}}(i) $ is the value that satisfies
	
	\begin{align*}
		|[f^{\leftarrow}_{\sigma,\sigma'}([a,b]), f^{\rightarrow}_{\sigma,\sigma'}([a,b])]_{B(\sigma)}| = 2\delta_{\vec{v'}}(i) + 1 ,
	\end{align*}
	where $ [a,b] \in B(\sigma') $ is the block in $ \sigma' $ that corresponds to entry $ i $ of $ \vec{v'} $.\footnote{
		The \nameref{def:horizon} $ \delta_{\vec{v'}}(i) $ is well defined because, by \Cref{claim:wrapping_block_interval_size_is_odd}, the number of blocks in $ [f^{\leftarrow}_{\sigma,\sigma'}([a,b]), f^{\rightarrow}_{\sigma,\sigma'}([a,b])]_{B(\sigma)} $ is odd.
	}
\end{definition}

\begin{observation}\label{claim:horizon_is_at_most_r}
	If $ \vec{v} $ is the \nameref{def:block_length_vector} corresponding to a \nameref{def:temporally_periodic} configuration $ \sigma $, then for every $ i \in \cycnums{k} $, where $ k $ is the length of $ \vec{v} $, the \nameref{def:horizon} $ \delta_{\vec{v}} $ satisfies $ \delta_{\vec{v}}(i) \le r $.
\end{observation}

\begin{claim}\label{claim:constant_horizon_for_temporally_periodic_paris}
	Let $ \vec{v}, \vec{v'} $ be a pair of \nameref{def:block_length_vector}s corresponding to a \nameref{def:temporally_periodic_pair} $ \sigma, \sigma' $.
	For every pair $ i,j \in \cycnums{|\vec{v}|} $,
	
	\begin{align*}
		\delta_{\vec{v}}(i) = \delta_{\vec{v}}(j) .
	\end{align*}
	Moreover, for every $ i \in \cycnums{|\vec{v'}|} $,
	\begin{align*}
		\delta_{\vec{v}}(i) = \delta_{\vec{v'}}(i) .
	\end{align*}
\end{claim}
\begin{proof}
	Let $ [a,b] \in B(\sigma) $ be the block corresponding to entry $ i $ of the \nameref{def:block_length_vector} $ \vec{v} $ and let $ [c,d] \in B(\sigma) $ be the block corresponding to entry $ j $ of the \nameref{def:block_length_vector} $ \vec{v} $.
	By \Cref{claim:temporally_periodic_intervals_have_constant_lengths}, the number of blocks in the block interval $ [f^{\leftarrow}_{\sigma',\sigma}([a,b]), f^{\rightarrow}_{\sigma',\sigma}([a,b])]_{B(\sigma')} $ equals the number of blocks in the block interval $ [f^{\leftarrow}_{\sigma',\sigma}([c,d]), f^{\rightarrow}_{\sigma',\sigma}([c,d])]_{B(\sigma')} $.
	That is, $ 2\delta_{\vec{v}}(i) + 1 = 2\delta_{\vec{v}}(j) + 1 $.
	Hence, $ \delta_{\vec{v}}(i) = \delta_{\vec{v}}(j) $, as claimed.
	
	The second part of the claim similarly follows from the second part of \Cref{claim:temporally_periodic_intervals_have_constant_lengths}.
\end{proof}

\begin{claim}\label{claim:length_any_sequence_of_2delta_blocks_at_most_2r}
	If $ \vec{v} $ is the \nameref{def:block_length_vector} corresponding to a \nameref{def:weakly_stable} configuration $ \sigma $, then the total length of every sequence of $ 2\delta $ \ref*{def:block}s in $ \sigma $ is at most $ 2r $, where $ \delta $ is the \nameref{def:horizon} of $ \vec{v} $ (which, by \Cref{claim:constant_horizon_for_temporally_periodic_paris}, is the same for all \ref*{def:block}s in the configuration $ \sigma $).
\end{claim}
\begin{proof}
	We show that for every sequence of $ 2\delta+1 $ consecutive \ref*{def:block}s in the configuration $ \sigma $, if we remove the leftmost \ref*{def:block} from the sequence, then the total length of the remaining \ref*{def:block}s is at most $ 2r $.
	
	Let $ \sigma' = \maj_r(\sigma) $.
	Since $ \sigma, \sigma' $ comprise a \nameref{def:temporally_periodic_pair} of \nameref{def:weakly_stable} configuration, every sequence of $ 2\delta+1 $ of consecutive \ref*{def:block}s in $ \sigma $ is of the form $ [f^{\leftarrow}_{\sigma,\sigma'}([a',b']), f^{\rightarrow}_{\sigma,\sigma'}([a',b'])]_{B(\sigma)} $ for some \ref*{def:block} $ [a',b'] \in B(\sigma') $.
	Let $ I $ be the \nameref{def:cell_interval} composed of the \ref*{def:block}s in the set $ [f^{\leftarrow}_{\sigma,\sigma'}([a',b']), f^{\rightarrow}_{\sigma,\sigma'}([a',b'])]_{B(\sigma)} \setminus \set{f^{\leftarrow}_{\sigma,\sigma'}([a',b'])} $.
	We prove that $ |I| \le 2r $.
	
	Let $ [a,b] = f^{\rightarrow}_{\sigma,\sigma'}([a',b']) $.
	That is, $ [a,b] $ is the \ref*{def:block} in $ B(\sigma) $ that contains the \ref{def:cell} $ a'+r $.
	By \Cref{claim:adjacent_blocks_in_temporally_periodic_pairs_are_mapped_into_adjacent_blocks}, if $ I_{b'+1} $ is the \ref*{def:block} that starts at the \ref{def:cell} $ b'+1 $ in $ \sigma' $, then $ f^{\rightarrow}_{\sigma,\sigma'}(I_{b'+1}) $ is the \ref*{def:block} that follows $ [a,b] $ in $ \sigma $, which means that $ b \in [a'+r, b'+1+r] $.
	Thus,
	
	\begin{align}
		|[a'+r,b]| \le b'+1+r - (a'+r) = |[a',b']| .
	\end{align}
	
	We conclude by bounding the length of the interval $ I $.
	
	\begin{align}
		|I| &\le [b'-r,b]\label{equation:I_contained_in_[b'-r,b]} \\
		&\le [b'-r,a'+r] + [a'+r,b] \\
		&\le (2r-|[a',b']|) + |[a',b']|\label{equation:assigning_interval_lengths} \\
		&= 2r ,
	\end{align}
	where Equation~(\ref{equation:I_contained_in_[b'-r,b]}) follows from the observation that $ I \subseteq [b'-r,b] $, and Equation~(\ref{equation:assigning_interval_lengths}) follows from noting that $ |[b'-r,a'+r]| = 2r-|[a',b']| $ and $ |[a'+r,b]| \le |[a',b']| $.
\end{proof}

\begin{claim}\label{claim:block_vector_length}
	If $ \vec{v} $ and $ \vec{v'} $ are a pair of \nameref{def:aligned} \nameref{def:block_length_vector}s of length $ k $ and \nameref{def:horizon} $ \delta $ corresponding to a \nameref{def:temporally_periodic_pair}, then for every $ i \in \cycnums{k} $,
	
	\begin{align*}
		\vec{v'}_i = \sum_{j=-\delta}^{\delta} (-1)^{j+\delta} \vec{v}_{i+j} .
	\end{align*}
\end{claim}
\begin{proof}
	Let $ \sigma, \sigma' $ be the \nameref{def:temporally_periodic_pair} that corresponds to the \nameref{def:block_length_vector} pair $ \vec{v}, \vec{v'} $ respectively.
	By \Cref{claim:block_length_for_temporally_periodic_configurations}, for $ \beta \in \set{0,1} $ and for every block $ [a',b'] \in B^{\beta}(\sigma') $,
	
	\begin{align*}
		|[a',b']| = \sum_{x \in A^{\beta}} |x| - \sum_{x \in A^{\bar{\beta}}} |x| ,
	\end{align*}
	where 
	\begin{align*}
		A^{\beta} &= [f^{\leftarrow}_{\sigma,\sigma'}([a',b']), f^{\rightarrow}_{\sigma,\sigma'}([a',b'])]_{B^{\beta}(\sigma)} , \\
		A^{\bar{\beta}} &= [f^{\leftarrow}_{\sigma,\sigma'}([a',b']), f^{\rightarrow}_{\sigma,\sigma'}([a',b'])]_{B^{\bar{\beta}}(\sigma)} .
	\end{align*}
	
	By \nameref{claim:switch_point_arg}, the value of the block $ f^{\leftarrow}_{\sigma,\sigma'}([a',b']) $ as well as the value of the block $ f^{\rightarrow}_{\sigma,\sigma'}([a',b']) $ is $ \beta $.
	In more detail, if $ j $ is the rightmost \ref{def:cell} of the block $ [a',b'] $, then the pair $ j,j+1 $ is a \nameref{def:switch_point} in the configuration $ \sigma' $, so by \nameref{claim:switch_point_arg}, $ \sigma(j-r) = \sigma'(j) $, which implies that the value of the block $ f^{\leftarrow}_{\sigma,\sigma'}([a',b']) $ is $ \beta $ because $ j-r \in f^{\leftarrow}_{\sigma,\sigma'}([a',b']) $.
	A similar argument holds for the block $ f^{\rightarrow}_{\sigma,\sigma'}([a',b']) $.
	
	The other block-values in the length $ 2\delta + 1 $ interval $ [f^{\leftarrow}_{\sigma,\sigma'}([a',b']), f^{\rightarrow}_{\sigma,\sigma'}([a',b'])]_{B(\sigma)} $ alternate between $ \bar{\beta} $ and $ \beta $.
	
	Hence, for every $ i \in \cycnums{k} $,
	
	\begin{align*}
		\vec{v'}_i = \sum_{j=-\delta}^{\delta} (-1)^{j+\delta} \vec{v}_{i+j} .
	\end{align*}
\end{proof}

\begin{claim}\label{claim:relating_two_temporally_periodic_length_vectors}
	If $ \vec{v} $ and $ \vec{v'} $ are a pair of \nameref{def:aligned} \nameref{def:block_length_vector}s corresponding to a \nameref{def:temporally_periodic_pair}, where the vectors are of length $ k $ and \nameref{def:horizon} $ \delta $, then for every $ i \in \cycnums{k} $,
	
	\begin{align*}
		\vec{v'}_i + \vec{v'}_{i+1} = \vec{v}_{i-\delta} + \vec{v}_{i+\delta+1} .
	\end{align*}
	Similarly,
	
	\begin{align*}
		\vec{v}_i + \vec{v}_{i+1} = \vec{v'}_{i-\delta} + \vec{v'}_{i+\delta+1} .
	\end{align*}
\end{claim}
\begin{proof}
	By \Cref{claim:block_vector_length}, for every $ i \in \cycnums{k} $,
	
	\begin{align*}
		\vec{v'}_i &= \sum_{j=-\delta}^{\delta} (-1)^{j+\delta} \vec{v}_{i+j} , \\
		\vec{v'}_{i+1} &= \sum_{j=-\delta}^{\delta} (-1)^{j+\delta} \vec{v}_{i+1+j} .
	\end{align*}
	Summing up the two equations we get:
	
	\begin{align*}
		\vec{v'}_i + \vec{v'}_{i+1} = \vec{v}_{i-\delta} + \vec{v}_{i+\delta+1} .
	\end{align*}
	And by symmetry,
	
	\begin{align*}
		\vec{v}_i + \vec{v}_{i+1} = \vec{v'}_{i-\delta} + \vec{v'}_{i+\delta+1} .
	\end{align*}
	The claim follows.
\end{proof}

\begin{claim}\label{claim:weakly_stable_configurations_are_balanced}
	If $ \sigma $ is a \nameref{def:temporally_periodic} configuration with block lengths at most $ r $ each, then $ \sigma $ is \nameref{def:balanced}.
\end{claim}
\begin{proof}	
	Let $ \sigma' = \maj_r(\sigma) $.
	Let $ \vec{v} $ and $ \vec{v'} $ be the corresponding \nameref{def:block_length_vector}s of the \nameref{def:temporally_periodic_pair} $ \sigma, \sigma' $.
	
	By \Cref{claim:temporally_periodic_vectors_have_same_length}, $ |\vec{v}| = \vec{v'} $.
	Denote that length by $ 2k $ (note that since configurations are cyclical, it is either the case that the configuration is homogeneous or that the length is even, and in the former case the claim trivially holds).
	
	Let
	
	\begin{align*}
		x = \sum_{i=0}^{k-1} \vec{v}_{2i} , \qquad
		y &= \sum_{i=0}^{k-1} \vec{v}_{2i+1} \\
		x' = \sum_{i=0}^{k-1} \vec{v'}_{2i} , \qquad
		y' &= \sum_{i=0}^{k-1} \vec{v'}_{2i+1} .
	\end{align*}
	By \Cref{claim:block_vector_length},
	
	\begin{align*}
		x &= (\delta+1)x' - \delta y' \\
		x &= (\delta+1)((\delta+1)x - \delta y) - \delta ((\delta+1)y - \delta x) \\
		x &= y .
	\end{align*}
	The claim follows.
\end{proof}

\section{The difference vectors of \nameref*{def:weakly_stable} configurations are \nameref*{def:spatially_periodic}}\label{sec:diff_vectors}

\begin{definition}\label{def:difference_vector}
	Let $ \vec{v} $ and $ \vec{v'} $ be a pair of \nameref{def:aligned} \nameref{def:block_length_vector}s of length $ k $ and \nameref{def:horizon} $ \delta $ corresponding to a \nameref{def:temporally_periodic_pair} of \nameref{def:weakly_stable} configurations.
	We define the pair of $ (\delta+1) $-steps difference vectors $ \Delta $ and $ \Delta' $ between $ \vec{v} $ and $ \vec{v'} $.
	For each $ i \in \cycnums{k} $,
	
	\begin{align*}
		\Delta_i &= \vec{v'}_{i+\delta+1} - \vec{v}_{i} , \\
		\Delta'_i &= \vec{v}_{i+\delta+1} - \vec{v'}_{i} .
	\end{align*}
\end{definition}

\begin{claim}\label{claim:sum_of_diff_vectors}
	If $ \vec{v} $ and $ \vec{v'} $ are a pair of \nameref{def:aligned} \nameref{def:block_length_vector}s of length $ k $ and \nameref{def:horizon} $ \delta $ corresponding to a \nameref{def:temporally_periodic_pair} of \nameref{def:weakly_stable} configurations, and $ \Delta $ and $ \Delta' $ are the pair's two $ (\delta+1) $-steps difference vectors, then for every $ i \in \cycnums{k} $ and every integer $ m \ge 0 $,
	
	\begin{align*}
		\sum_{j=0}^{m-1} \Delta_{i+2j(\delta+1)} + \sum_{j=0}^{m-1} \Delta'_{i+(2j+1)(\delta+1)} = \vec{v}_{i+2m(\delta+1)} - \vec{v}_i .
	\end{align*}
\end{claim}
\begin{proof}
	We prove the claim by induction on $ m $.
	For $ m=0 $, the sums in the left-hand side are empty and the right-hand side consist of a difference between two equal terms, so the equation clearly holds.
	We assume as our induction hypothesis that the equation holds for $ m $, and prove it for $ m+1 $.
	
	\begin{align}
		\sum_{j=0}^{m} \Delta_{i+2j(\delta+1)} + \sum_{j=0}^{m} \Delta'_{i+(2j+1)(\delta+1)} &= \sum_{j=0}^{m-1} \Delta_{i+2j(\delta+1)} + \sum_{j=0}^{m-1} \Delta'_{i+(2j+1)(\delta+1)} \\
		&\quad + \Delta_{i+2m(\delta+1)} + \Delta'_{i+(2m+1)(\delta+1)} \\
		&= \vec{v}_{i+2m(\delta+1)} - \vec{v}_i + \Delta_{i+2m(\delta+1)} + \Delta'_{i+(2m+1)(\delta+1)} \label{equation:applying_induction_hypothesis} \\
		&= (\vec{v'}_{i+(2m+1)(\delta+1)} - \Delta_{i+2m(\delta+1)})\label{equation:apply_diff_vector_definition_for_v} \\
		&\quad - \vec{v}_i + \Delta_{i+2m(\delta+1)} + \Delta'_{i+(2m+1)(\delta+1)} \\
		&= (\vec{v}_{i+(2m+2)(\delta+1)} - \Delta'_{i+(2m+1)(\delta+1)}) - \Delta_{i+2m(\delta+1)}\label{equation:apply_diff_vector_definition_for_v'} \\
		&\quad - \vec{v}_i + \Delta_{i+2m(\delta+1)} + \Delta'_{i+(2m+1)(\delta+1)} \\
		&= \vec{v}_{i+(2m+2)(\delta+1)} - \vec{v}_i ,
	\end{align}
	where Equation~(\ref{equation:applying_induction_hypothesis}) follows from applying the induction hypothesis, and \nameCrefs{equation:apply_diff_vector_definition_for_v} (\ref{equation:apply_diff_vector_definition_for_v}) and (\ref{equation:apply_diff_vector_definition_for_v'}) follow directly from applying \Cref{def:difference_vector} to the \nameref{def:block_length_vector} $ \vec{v}_{i+2m(\delta+1)} $ and then to the \nameref{def:block_length_vector} $ \vec{v'}_{i+(2m+1)(\delta+1)} $.
\end{proof}

\begin{claim}\label{claim:steps_diff_vectors_are_spatially_periodic}
	If $ \vec{v} $ and $ \vec{v'} $ are a pair of \nameref{def:aligned} \nameref{def:block_length_vector}s of length $ k $ and \nameref{def:horizon} $ \delta $ corresponding to a \nameref{def:temporally_periodic_pair} of \nameref{def:weakly_stable} configurations, then each of the pair's two $ (\delta+1) $-steps difference vectors $ \Delta $ and $ \Delta' $ are \nameref{def:spatially_periodic} with a \nameref{def:spatial_period} that divides $ 2\delta $.
	In other words, for every $ i \in \cycnums{k} $,
	
	\begin{align*}
		\Delta_i &= \Delta_{i+2\delta} , \\
		\Delta'_i &= \Delta'_{i+2\delta} .
	\end{align*}
\end{claim}
\begin{proof}
	By \Cref{claim:relating_two_temporally_periodic_length_vectors}, for every $ i \in \cycnums{k} $,
	
	\begin{align*}
		\vec{v'}_i + \vec{v'}_{i+1} = \vec{v}_{i-\delta} + \vec{v}_{i+\delta+1} .
	\end{align*}
	Hence,
	
	\begin{align*}
		\vec{v}_{i+\delta+1} - \vec{v'}_i = \vec{v'}_{i+1} - \vec{v}_{i-\delta} .
	\end{align*}
	By definition, $ \Delta'_i = \vec{v}_{i+\delta+1} - \vec{v'}_{i} $ and $ \Delta_{i-\delta} = \vec{v'}_{i+1} - \vec{v}_{i-\delta} $.
	Thus,
	
	\begin{align}\label{eq:expressing_Delta'_in_terms_of_Delta}
		\Delta'_i &= \Delta_{i-\delta} .
	\end{align}
	Similarly, also by \Cref{claim:relating_two_temporally_periodic_length_vectors}, for every $ i \in \cycnums{k} $,
	
	\begin{align*}
		\vec{v}_i + \vec{v}_{i+1} = \vec{v'}_{i-\delta} + \vec{v'}_{i+\delta+1} ,
	\end{align*}
	which means that
	
	\begin{align*}
		\vec{v'}_{i+\delta+1} - \vec{v}_i = \vec{v}_{i+1} - \vec{v'}_{i-\delta} 
	\end{align*}
	as well, and since $ \Delta_i = \vec{v'}_{i+\delta+1} - \vec{v}_{i} $ and $ \Delta'_{i-\delta} = \vec{v}_{i+1} - \vec{v'}_{i-\delta} $ we also get
	
	\begin{align}\label{eq:expressing_Delta_in_terms_of_Delta}
		\Delta_i &= \Delta'_{i-\delta} .
	\end{align}
	
	Combining equations (\ref{eq:expressing_Delta'_in_terms_of_Delta}) and (\ref{eq:expressing_Delta_in_terms_of_Delta}) for $ i+2\delta $ and for $ i+\delta $,
	
	\begin{align*}
		\Delta_{i+2\delta} &= \Delta'_{i+\delta} = \Delta_i , \\
		\Delta'_{i+2\delta} &= \Delta_{i+\delta} = \Delta'_i .
	\end{align*}
	That is, the \nameref{def:spatial_period} of each of the vectors $ \Delta $ and $ \Delta' $ is at most $ 2\delta $.
\end{proof}

\begin{claim}\label{claim:sum_of_diff_vectors_is_constant}
	Let $ \vec{v} $ and $ \vec{v'} $ be a pair of length-$ k $ \nameref{def:horizon}-$ \delta $ \nameref{def:aligned} \nameref{def:block_length_vector}s corresponding to a \nameref{def:temporally_periodic_pair} of \nameref{def:weakly_stable} configurations, and let $ \Delta $ and $ \Delta' $ be the pair's two $ (\delta+1) $-steps difference vectors.
	For every $ i \in \cycnums{k} $, define
	
	\begin{align*}
		\varsigma(i) = \sum_{j=0}^{2\delta-1} \Delta_{i+2j(\delta+1)} + \sum_{j=0}^{2\delta-1} \Delta'_{i+(2j+1)(\delta+1)} .
	\end{align*}
	Then $ \varsigma(i) = \varsigma(i') $ for every pair $ i,i' \in \cycnums{k} $.
\end{claim}
\begin{proof}
	We show that for every $ i \in \cycnums{k} $, the value of $ \varsigma(i) $ does not depend on $ i $.
	
	\begin{align}
		\varsigma(i) &= \sum_{j=0}^{2\delta-1} \Delta_{i+2j(\delta+1)} + \sum_{j=0}^{2\delta-1} \Delta'_{i+(2j+1)(\delta+1)} \\
		&= \sum_{j=0}^{2\delta-1} \Delta_{i+2j} + \sum_{j=0}^{2\delta-1} \Delta'_{i+2\delta+2j+1}\label{equation:apply_spatial_periodicity_of_diff_vector} \\
		&= \sum_{j=0}^{2\delta-1} \Delta_{i+2j} + \sum_{j=0}^{2\delta-1} \Delta_{i+2j+1}\label{equation:replace_Delta'_by_Delta} \\
		&= \sum_{j=0}^{2\delta-1} \Delta_{i+j} \label{equation:combine_the_two_sums} \\
		&= \sum_{j=0}^{2\delta-1} \Delta_{j} , \label{equation:get_rid_of_i} ,
	\end{align}
	where:
	\begin{itemize}
		\item[--] Equation~(\ref{equation:apply_spatial_periodicity_of_diff_vector}) follows from noting that, by \Cref{claim:steps_diff_vectors_are_spatially_periodic}, the vector $ \Delta $ is \nameref{def:spatially_periodic} with a \nameref{def:spatial_period} that divides $ 2\delta $, so $ \Delta_{i+2j(\delta+1)} = \Delta_{i+2j} $ and
		$ \Delta'_{i+(2j+1)(\delta+1)} = \Delta'_{i+\delta+2j+1} $ for every $ 0 \le j \le 2\delta-1 $.
		
		\item[--] Equation~(\ref{equation:replace_Delta'_by_Delta}) follows from the observation that, by \Cref{eq:expressing_Delta'_in_terms_of_Delta} in the proof of \Cref{claim:steps_diff_vectors_are_spatially_periodic}, $ \Delta'_{i+\delta+2j+1} = \Delta_{i+2j+1} $ for every $ 0 \le j \le m-1 $.
		
		\item[--] Equation~(\ref{equation:combine_the_two_sums}) follows from combining the two summations.
		
		\item[--] Equation~(\ref{equation:get_rid_of_i}) follows again from the property that the vector $ \Delta $ is \nameref{def:spatially_periodic} with a \nameref{def:spatial_period} that divides $ 2\delta $ (\Cref{claim:steps_diff_vectors_are_spatially_periodic}).
	\end{itemize}
\end{proof}

\section{The \nameref*{def:weakly_stable} configurations are \nameref*{def:spatially_periodic}}\label{sec:spatially_periodic}

\begin{claim}\label{claim:spatially_periodic}
	The \nameref{def:spatial_period} of every \nameref{def:weakly_stable} configuration is at most $ 2r(r+1) $.
\end{claim}
\begin{proof}
	Let $ \sigma, \sigma' $ be a \nameref{def:temporally_periodic_pair} of \nameref{def:weakly_stable} configurations, and let $ \vec{v} $ and $ \vec{v'} $ be the corresponding \nameref{def:aligned} \nameref{def:block_length_vector}s.
	Let $ k $ be the length of $ \vec{v} $ and $ \vec{v'} $ and let $ \delta $ be their \nameref{def:horizon}.
	
	Let $ \Delta $ and $ \Delta' $ be the two $ (\delta+1) $-steps difference vectors of the pair $ \vec{v}, \vec{v'} $.
	By \Cref{claim:sum_of_diff_vectors}, for every integer $ m \ge 0 $,
	
	\begin{align}\label{eq:m_steps_from_v_i}
		\vec{v}_i + \sum_{j=0}^{m-1} \Delta_{i+2j(\delta+1)} + \sum_{j=0}^{m-1} \Delta'_{i+(2j+1)(\delta+1)} = \vec{v}_{i+2m(\delta+1)} .
	\end{align}
	If we assign $ m = 2\delta $, by \Cref{claim:sum_of_diff_vectors_is_constant}, the expression consisting of the two sums in Equation~(\ref{eq:m_steps_from_v_i}) does not depend on $ i $, so we denote that expression by $ \varsigma $.
	That is, for every $ i \in \cycnums{k} $,
	
	\begin{align}
		\vec{v}_i + \varsigma = \vec{v}_{i+2\delta(\delta+1)} ,
	\end{align}
	
	Since $ \vec{v} $ is a \nameref{def:block_length_vector}, the values in its entries are bounded, and so it must be the case that $ \varsigma = 0 $.
	Thus, for every $ i \in \cycnums{k} $,
	
	\begin{align*}
		\vec{v}_i &= \vec{v}_{i+2\delta(\delta+1)} .
	\end{align*}
	That is, the \nameref{def:block_length_vector} $ \vec{v} $ is \nameref{def:spatially_periodic} with \nameref{def:spatial_period} at most $ 2\delta(\delta+1) $.
	By \Cref{claim:length_any_sequence_of_2delta_blocks_at_most_2r}, the total length of every sequence of $ 2\delta $ \ref*{def:block}s in $ \sigma $ is at most $ 2r $, and, by \Cref{claim:horizon_is_at_most_r}, $ \delta \le r $, implying that the configuration $ \sigma $ is \nameref{def:spatially_periodic} with \nameref{def:spatial_period} at most $ 2r(r+1) $.
\end{proof}

\section{Putting it all together: Proving \Cref*{thm:characterization}}\label{sec:putting_it_all_together}

\begin{proof}[Proof of \Cref{thm:characterization}]
	Let $ \sigma : \cycnums{n} \to \bitset $ be any configuration, let $ \sigma' = \maj_r(\sigma) $ and let $ \sigma'' = \maj_r(\sigma') $.
	
	Suppose first that the configuration $ \sigma $ is \nameref{def:temporally_periodic}.
	If all the \ref{def:cell}s in $ \sigma $ are \nameref{def:strongly_stable}, then $ \sigma $ is by definition of the form $ (0^{r+1}0^* + 1^{r+1}1^*)^* $, and by \Cref{claim:strongly_stable_configurations_are_fixed_points}, it is a \nameref{def:fixed_point}.
	Otherwise, still assuming that $ \sigma $ is \nameref{def:temporally_periodic}, it must be the case that $ \sigma $ is \nameref{def:weakly_stable}, which means that, by \Cref{claim:weakly_stable_configurations_have_no_strongly_stable_cells}, all the \ref{def:cell}s in $ \sigma $ are \nameref{def:weakly_stable}, and by \Cref{claim:spatially_periodic}, the configuration $ \sigma $ is \nameref{def:spatially_periodic} with \nameref{def:spatial_period} at most $ 2r(r+1) $.
	
	Suppose now that the configuration $ \sigma $ is \nameref{def:transient}.
	By \Cref{claim:temporally_periodic_blocks_are_short}, the length of every \ref*{def:block} in $ B(\sigma) $ is at most $ r $.
	We say that a \nameref{def:cell_interval} $ [i,j] \subseteq \cycnums{n} $ is \nameref{def:unstable} with respect to $ \sigma $ if for every $ \ell \in [i,j] $, $ \sigma(\ell) \ne \sigma''(\ell) $.
	We claim that every \nameref{def:unstable} \nameref{def:cell_interval} with respect to $ \sigma $ contains at most one \nameref{def:switch_point}.
	
	Suppose by way of contradiction that there exists an \nameref{def:unstable} \nameref{def:cell_interval} with respect to $ \sigma $ that contains at least two \nameref{def:switch_point}s.
	Let $ (i,i+1) $ and $ (j-1,j) $ be the first two \nameref{def:switch_point}s in the \nameref{def:unstable} \nameref{def:cell_interval} (without loss of generality, $ (i,i+1) $ is to the left of $ (j-1,j) $ and there is no other \nameref{def:switch_point} between them).
	Hence, for some $ \beta \in \set{0,1} $, it holds that $ \sigma(i) = \sigma(j) = \beta $, and for every $ \ell \in (i,j) $, $ \sigma(\ell) = \bar{\beta} $.
	
	Since $ [i,j] $ is also an \nameref{def:unstable} \nameref{def:cell_interval}, it must be the case that $ \sigma''(i) = \sigma''(i+1) = \bar{\beta} $ and for every $ \ell \in (i,j) $, $ \sigma''(\ell) = \beta $.
	Hence, both $ (i,i+1) $ and $ (j-1,j) $ are \nameref{def:switch_point}s in $ \sigma'' $.
	By \nameref{claim:switch_point_arg}, $ \sigma'(i+1+r)=\beta $ and $ \sigma'(j+r)=\bar{\beta} $.
	
	Let $ \ell' $ be the rightmost \ref{def:cell} in the open interval $ (i+r,j+r) $ where $ \sigma'(\ell') = \beta $.
	Since the pair $ (\ell',\ell'+1) $ constitutes a \nameref{def:switch_point} in $ \sigma' $, by \nameref{claim:switch_point_arg}, it must hold that $ \sigma(\ell'-r) = \beta $.
	However, since $ \ell'-r \in (i,j) $ (because $ \ell' \in (i+1,j+r) $), we reach a contradiction to the assumption that for every $ \ell \in (i,j) $, it must hold that $ \sigma(\ell) = \bar{\beta} $.
	
	Thus, every \nameref{def:unstable} \nameref{def:cell_interval} with respect to $ \sigma $ contains at most one \nameref{def:switch_point}, which implies that the maximum possible length of an \nameref{def:unstable} \nameref{def:cell_interval} is $ 2r $.
\end{proof}

\newpage

\bibliographystyle{alpha}
\bibliography{refs}

\end{document}